\documentclass[11pt,oneside,letter,reqno]{amsart}

\usepackage{amsthm, amsmath, amssymb, amsfonts, graphicx, epsfig}
\usepackage{accents}
\usepackage{bbm}
\usepackage{mathrsfs,dsfont}
\usepackage[normalem]{ulem}

\usepackage{cancel}

\usepackage{epstopdf}

\usepackage{graphicx}
\usepackage[font=sl,labelfont=bf]{caption}
\usepackage{subcaption}

\usepackage{float}
\restylefloat{table}

\usepackage{enumerate}

\usepackage[colorlinks=true, pdfstartview=FitV, linkcolor=blue,
            citecolor=blue, urlcolor=blue]{hyperref}
\usepackage[usenames]{color} 

\usepackage[round]{natbib}

\usepackage{url}

\makeatletter\def\url@leostyle{%
 \@ifundefined{selectfont}{\def\UrlFont{\sf}}{\def\UrlFont{\scriptsize\ttfamily}}} \makeatother

\newtheorem{theorem}{Theorem}
\newtheorem{proposition}[theorem]{Proposition}

\theoremstyle{definition}
\newtheorem{definition}[theorem]{Definition}
\newtheorem{example}[theorem]{Example}
\theoremstyle{remark}
\newtheorem{remark}[theorem]{Remark}

\numberwithin{equation}{section}
\numberwithin{theorem}{section}

\definecolor{Red}{rgb}{0.8,0,0.1}

\def\cA{\mathcal{A}}

\def\cD{\mathcal{D}}

\def\cN{\mathcal{N}}

\def\bE{\mathbb{E}}

\def\bN{\mathbb{N}}

\def\bP{\mathbb{P}}

\def\bR{\mathbb{R}}

\newcommand{\1}{\mathbbm{1}}            

\DeclareMathOperator{\var}{\mathrm{V}@\mathrm{R}}           

\newcommand{\Ind}{{\mathds 1}}
\newcommand{\ind}[1]{\Ind_{\{#1\}}}

\makeatletter
\def\namedlabel#1#2{\begingroup
    #2%
    \def\@currentlabel{#2}%
    \phantomsection\label{#1}\endgroup
}
\makeatother

\title{Unbiased estimation of risk}

\usepackage[backgroundcolor=white,bordercolor=orange]{todonotes}
\usepackage{booktabs,nicefrac}
\usepackage{array}

\author{
        Marcin Pitera \and Thorsten Schmidt}
 \address{Institute of Mathematics, Jagiellonian University, \L{}ojasiewicza 6, 30-348 Cracow, Poland} \email{\url{marcin.pitera@im.uj.edu.pl}}
 \address{Department of Mathematical Stochastics, University of Freiburg, Eckerstr.1, 79104 Freiburg, Germany}
   \email{\url{thorsten.schmidt@stochastik.uni-freiburg.de} }

\RequirePackage{colortbl}
\definecolor{ts}{rgb}{0,0.6,0}

\binoppenalty=\maxdimen 
\relpenalty=\maxdimen

\date{First circulated: March 8, 2016, This version: \today. The authors express their gratitude to the referees and the associate editor for valuable comments and suggestions that helped to improve the paper. Part of the work of the first author was supported by the National Science Centre, Poland, via project 2016/23/B/ST1/00479.}

\newcommand{\R}{\bR}
\newcommand{\N}{\bN}

\defcitealias{Bas2013}{BCBS, 2013} 
\defcitealias{Bas2011}{BCBS, 2011}
\defcitealias{Bas2006}{BCBS, 2006} 
\defcitealias{Bas1996}{BCBS, 1996}

\begin{document}
\maketitle

\begin{abstract} 

The estimation of risk measures recently gained  a lot of attention, partly because of the backtesting issues of expected shortfall related to elicitability. In this work we shed a new and fundamental light on optimal estimation procedures of risk measures in terms of bias.
We show that once the parameters of a model need to be estimated, one has to take additional care when estimating risks. 
The typical plug-in approach, for example, introduces a bias which leads to a systematic \emph{underestimation} of risk. 

In this regard, we introduce a novel notion of \emph{unbiasedness} to the estimation of risk which is motivated by economic principles. In general,  the proposed concept does not coincide with the well-known statistical notion of unbiasedness. We show that an appropriate bias correction is available for many well-known estimators.  In particular, we consider value-at-risk and expected shortfall (tail value-at-risk). In the special case of normal distributions, closed-formed solutions for unbiased estimators can be obtained. 

We present a number of motivating examples which show the outperformance of unbiased estimators in many circumstances. The unbiasedness has a direct impact on backtesting and therefore adds a further viewpoint to established statistical properties. 
\end{abstract}

\vspace{2mm}


\section{Introduction}\label{S:introduction}
The estimation of risk measures  is an area of highest importance in the financial industry as risk measures play a major role in risk management and in the computation of regulatory capital, see \cite{MFE} for an in-depth treatment of the topic. Most notably, a major part of quantitative risk management is  of statistical nature, as highlighted for example in \cite{embrechts2014}. This article takes this challenge seriously and does not target risk measures themselves, but \emph{estimated} risk measures. 

Statistical aspects in the estimation of risk measures recently raised a lot of attention: see the related articles \cite{Davis2016} and \cite{ConDegSca2010,Acerbi2014Risk,Ziegel2014, Fissler2016,Frank2016}. Surprisingly, it turns out that statistical properties of risk estimators - related to the presence of bias - have not yet been analysed thoroughly. Such properties are very important from the practical point of view, as the risk bias usually leads to a systematic underestimation of risk. It is our main goal to give a definition of \emph{unbiasedness} that makes sense economically and statistically. The main motivation for this is the observation that the classical (statistical) definition of bias might be  desirable from a theoretical point of view, while it might be not prioritised by financial institutions or regulators, for whom the backtests are currently the major source of evaluating the estimation. 

There is an ongoing intensive debate in regulation and in science about the two most recognised risk measures: Expected Shortfall (ES) and Value-at-Risk ($\var$). This debate is stimulated by Basel III project \citepalias{Bas2013}, which updates regulations responsible for capital requirements for initial market risk models (cf. \citepalias{Bas2006,Bas1996}). In a nutshell, the old $\var$ at level $1\%$ is replaced with ES at level $2.5\%$. In fact, such a correction may reduce the bias, however only in the right scenarios. The academic response to this fact is not unanimous: while ES is coherent and takes into consideration the whole tail distribution, it lacks some nice statistical properties characteristic to $\var$. See e.g.~\cite{ConDegSca2010,Acerbi2014Risk,Ziegel2014,KelRos2016,YamYos2005,Emmeretal2015} for further details and interesting discussions. Also, the ES forecasts are believed to be much harder to backtest, a property essential from the regulator's point of view.

A further argument in this debate emerges from the results in \cite{Gneiting2011} (see also \cite{Web2006}), showing that ES is not \emph{elicitable}. This interesting concept was originally developed in \cite{osband1985information} from an economic perspective; the main motivation is to ensure truthful reporting by penalizing false reports. In Section \ref{sec:elicitabilty} we provide a detailed discussion of this topic. 
The lack of elicitability led to the discussion whether or not (and how) it is possible to backtest ES and we refer to \cite{Car2014,Ziegel2014, Acerbi2014Risk,Fissler2016} for further details on this topic. Quite recently it was shown in~\cite{Fissler2016}  that ES is however jointly elicitable with $\var$.\footnote{Another illustrative and self-explanatory example of this phenomena is variance. While not being elicitable, it is jointly elicitable with the mean; see~\cite{LamPenSho2008}.}
In particular, backtesting ES is possible; see Section \ref{sec5.3} for a backtesting algorithm of ES in our setting; however, the results have to be interpreted with care. 

Our article has two objectives. First, we introduce an economically motivated definition of unbiasedness: 
an estimation of risk capital is called \emph{unbiased}, if adding the estimated amount  of risk capital to the risky position makes the position acceptable; see Definition \ref{def:unbiased}. It seems to be surprising that this is not the case for estimators considered so far. Second, we want to shed a new light on backtesting starting from the viewpoint of the standard Basel requirements. The starting point is the simple observation that a biased estimation  naturally leads to a poor performance in backtests, such that the suggested bias correction should improve the results in backtesting.

In this regard, consider the standard (regulatory) backtest for $\var$ which is based on the rate of exception; see~{\cite{GioLau2003}. 
In Section ~\ref{sec:backtesting} we will show that a backtesting procedure based on the rate of exception will perform poorly if the estimation of the rate of exception is biased. Motivated by this, we systematically study bias of risk estimators and link our theoretical foundation to empirical evidence.

 Let us start with an example: consider i.i.d.\ Gaussian data with unknown mean and variance and assume we are interested in estimating $\var$ at the level $\alpha\in(0,1)$ ($\var_\alpha$). Denote by $x=(x_1,\dots,x_n)$ the observed data. The \emph{unbiased} estimator in this case  is given by
\begin{equation}
\hat{\var}_{\alpha}^{\textrm{u}}(x_1,\ldots,x_n):=-\left(\bar{x} +\bar{\sigma}(x)\sqrt{\frac{n+1}{n}}t_{n-1}^{-1}(\alpha)\right)\label{eq:var.bia},
\end{equation}
where $t_{n-1}^{-1}$ is the inverse of the cumulative distribution function of the Student-$t$-distribution with $n-1$ degrees of freedom, $\bar x$ denotes the sample mean and $\bar \sigma(x)$ denotes the sample standard deviation. 
We call this estimator the \emph{Gaussian unbiased estimator} and use this name throughout as reference to \eqref{eq:var.bia}. Note that the $t$-distribution arises naturally by taking into account that variance has to be estimated and that the bias correction factor is $\sqrt{\nicefrac{n+1}{n}}$.
Comparing this estimator to standard estimators on NASDAQ data provides some motivating insights which we detail in the following paragraph.

\subsubsection*{Backtesting value-at-risk estimating procedures.}\label{S:example}
To analyse the performance of various estimators of value-at-risk we performed a standard backtesting procedure. First, we estimated the risk measures using a learning period and then tested their adequacy in the backtesting period. The test was based on the standard {\it failure rate} ({\it exception rate}) procedure; see e.g. \cite{GioLau2003} and \citepalias{Bas1996}. Given a data sample of size $n$, the first $k$ observations were used for estimating the value-at-risk at level $\alpha$. Afterwards it was counted how many times the actual loss in the following $n-k$ observations exceeded the estimate. For good estimators, we would expect that the number of exceptions divided by $(n-k)$ should be close to $\alpha$.

More precisely, we considered returns based on (adjusted) closing prices of the NASDAQ100 index in the period from 1999-01-01 to 2014-11-25. The sample size is $n=4000$, which corresponds to the number of trading days in this period. The sample was split into 80 separate subsets, each consisting of the consecutive 50 trading days. The  backtesting procedure consisted in using the $i$-th subset for estimating the value of $\var_{0.05}$ and counting the number of exceptions in the $(i+1)$-th subset. The total number of exceptions in the 79 periods was divided by $79\cdot 50$. We compared the performance of the Gaussian unbiased estimator $\hat{\var}_{\alpha}^{\textrm{u}}$ to the three most common estimators of value-at-risk:
the empirical sample quantile $\hat{\var}_{\alpha}^{\textrm{emp}}$ (sometimes called  historical estimator\footnote{In fact there are numerous versions of the sample quantile estimator. We have decided to take the one  used by default both in {\bf R} and {\bf S} statistical software for samples from continuous distribution.}); the  modified Cornish-Fisher estimator $\hat{\var}_{\alpha}^{\textrm{CF}}$; and the classical Gaussian estimator $\hat{\var}_{\alpha}^{\textrm{norm}}$, which is obtained by inserting mean and sample variance into the value-at-risk formula under normality:  
\begin{align}
\hat{\var}_{\alpha}^{\textrm{emp}}(x) &:=-\left(x_{(\lfloor h\rfloor)}+(h-\lfloor h\rfloor)(x_{(\lfloor h+1\rfloor)}-x_{(\lfloor h \rfloor)}) \right),\label{eq:var.hist}\\
\hat{\var}_{\alpha}^{\textrm{CF}}(x) & :=-\left(\bar{x} +\bar{\sigma}(x)\bar{Z}^{\alpha}_{CF}(x)\right),\label{eq:var.mod}\\
\hat{\var}_{\alpha}^{\textrm{norm}}(x) & :=-\left(\bar{x}+\bar{\sigma}(x)\Phi^{-1}(\alpha)\right),\label{eq:var.norm}
\end{align}
where $x_{(k)}$ is the $k$-th order statistic of  $x=(x_1,\ldots,x_n)$, the value $\lfloor z \rfloor$ denotes the integer part of $z\in\bR$, $h=\alpha(n-1)+1$,  $\Phi$ denotes the cumulative distribution function of the standard normal distribution and $\bar{Z}^{\alpha}_{CF}$ is a standard Cornish-Fisher $\alpha$-quantile estimator (see e.g.~\cite[Section IV.3.4.3]{Car2009} for details).
\begin{table}[t]
\caption{Estimates of $\var_{0.05}$ for NASDAQ100 (first column) and for a sample from normally distributed random variable with mean and variance fitted to the NASDAQ data (second column), both for 4.000 data points. \emph{Exceeds} reports the number of exceptions in the sample, where the actual loss exceeded the risk estimate. The expected rate of  $0.05$ is only reached by the Gaussian unbiased estimator.}
\centering
\begin{tabular}{ll*{3}{c}c}\toprule
 Estimator&&  \multicolumn{2}{c}{NASDAQ} & \multicolumn{2}{c}{Simulated} \\
 &&  exceeds & percentage &  exceeds & percentage  \\ \midrule
Gaussian plug-in & $\hat{\var}_{\alpha}^{\textrm{norm}}$ & 241 & 0.061 & 221 & 0.056\\
Empirical & $\hat{\var}_{\alpha}^{\textrm{emp}}$ & 272 & 0.069 & 253 & 0.064 \\
Modified C-F & $\hat{\var}_{\alpha}^{\textrm{CF}}$ & 249 & 0.063 & 230 &  0.058\\
Gaussian unbiased& $\hat{\var}_{\alpha}^{\textrm{u}}$ & 217 & 0.055 & 197 & 0.050 \\ \bottomrule
\end{tabular}
\label{t:test1}
\end{table}

The results of the backtest are shown in the first part of Table~\ref{t:test1}. Surprisingly, the standard estimators show a rather poor performance. Indeed, one would expect a failure rate of 0.05 when using an estimator for the $\var_{0.05}$ and the standard estimators show a clear \emph{underestimation} of the risk, i.e.\ an exceedance rate higher than the expected rate. Only  the Gaussian  unbiased estimator is close to the expected rate, the empirical estimator having an exceedance rate which is 
25\% higher in comparison.  One can also show that a Student-$t$ (plug-in) estimator performs poorly, compared to Gaussian unbiased estimator.

To exclude possible disturbances of these findings by a bad fit of the Gaussian model to the data or possible dependences we additionally performed a simulation study: starting from an i.i.d. sample of normally distributed random variable with mean and variance fitted to the NASDAQ data we repeated the backtesting on this data; results are shown in the second column of Table~\ref{t:test1}. Let us first focus on the plug-in estimator $\hat{\var}_{\alpha}^{\textrm{norm}}$: expecting 200 exceedances (5\% out of 4.000) we experienced additional 41 exceedances on the NASDAQ data itself. On the simulated data, where we can exclude disturbances due to fat tails, correlation etc., still 21 unexpected exceedances were reported which is roughly 50\% of the additional exceedances on the original data. These exceedances are due to the biasedness of the estimator and can be removed by considering the unbiased estimator $\hat{\var}_{\alpha}^{\textrm{u}}$ as may be seen from the last line of Table~\ref{t:test1}. 
The results on the other estimators confirm these findings\footnote{Further simulations show that these results are stastically significant: for example, repeating the simulation 10.000 times 
allows to compute the mean exception rates (with standard errors in parentheses) for estimators given in Table~\ref{t:test1}. They are equal to 0.057 (0.0026), 0.067 (0.0028), 0.057 (0.0028),  and 0.050 (0.0027), respectively.}, the empirical estimator shows an exceedance rate being 28\% higher compared to the Gaussian unbiased estimator, which in turn perfectly meets the level $\alpha=0.05$. However, as has been pointed out in \cite{Gneiting2011}, conclusions about the performance of estimators solely based on backtesting have to be taken with care and we provide additional evidence in Section \ref{sec:nonstandardbacktesting}.

The structure of the paper is as follows: in Section \ref{S:EstimationRisk}, estimators of risk measures are formally introduced. Section \ref{S:Plug-in} discusses the frequently used concept of plug-in estimators. Section \ref{S:Robust} introduces the main concept of the paper, \emph{unbiasedness} in an economic sense, and Section~\ref{S:examples} gives a number of examples. Section \ref{sec:asymptotics} considers asymptotically unbiased estimators, while Section~\ref{sec:backtesting} outlines the bias estimation procedure as well as the relation between unbiasedness and regulatory backtesting. Section \ref{sec:empirics} gives a small empirical study of the proposed estimators and we conclude in Section \ref{sec:conclusion}.

\section{Measuring risk}\label{S:EstimationRisk}
In this section we give a short introduction to risk measures where the underlying model is not known and hence needs to be estimated. Our focus lies on the most popular family of risk measures, so-called \emph{law-invariant} risk measures. These measures solely depend on the distribution of the underlying losses, see \cite{MFE} for an outline and practical applications of risk measurement. Law-invariant risk measures for example contain the special cases value-at-risk, expected shortfall or spectral risk measures. 

We consider the estimation problem in a parametric setup. If the parameter space is chosen infinite-dimensional, this also contains the non-parametric formulation of the problem.
In this regard, let $(\Omega,\cA)$ be a measurable space and $(P_\theta:\theta \in \Theta)$ be a family of probability measures on this measurable space parametrized by $\theta$, an element of the parameter space $\Theta$. 
For simplicity, we assume that the measures $P_\theta$ are equivalent, such that their null-sets coincide. Otherwise it could be possible that some of the probability measures would be excluded almost surely by some observation which in turn would lead to unnecessarily complicated expressions. By  $L^{0}:=L^{0}(\Omega,\cA)$ we denote the (equivalence classes of) real-valued and measurable functions. In our context, the space $L^{0}$ typically corresponds to discounted cash flows or financial positions return rates.

For the estimation, we assume that we have a sample $X_1,X_2,\dots,X_n$ of observations at hand and we want to estimate the risk of a future position $X$.  Sometimes, we consider $x=(x_1,\dots,x_n)$  to distinguish specific realizations of $X_1,\dots,X_n$ from the sample random variables. In particular, we know that $x_{i}=X_{i}(\omega)$ for some $\omega\in\Omega$.

\begin{example}[The i.i.d.-case]\label{ex:iid}
Assume that the future position $X$ as well as the historical observations are independent and identically distributed (i.i.d.). This is the case, for example in the Black-Scholes model when one considers log-returns. More generally, this also holds in the case where the considered stock price $S$ follows a geometric L\'evy process (see~\cite[Section 5.6.2]{App2009} and references therein). If $t_i,$ $i=0,\dots,n+1$ denote equidistant times with $\Delta=t_i-t_{i-1}$, then the log-returns $X_i := \log (S_{t_i}) - \log(S_{t_{i-1}})$, $i=1,\dots,n$ are i.i.d. and the risk of the future position $S_{t_{n+1}}$ can be described in terms of $X:=\log (S_{t_{n+1}}) - \log(S_{t_{n}})$.
\end{example}

To quantify the risk associated with a future position $X\in L^0$, we introduce a concept of {\it a risk measure}.
\begin{definition}\label{def:risk.measure}
A risk measure $\rho$ is a mapping from $L^0$ to $\bR\cup \{+\infty\}$.
\end{definition}
Typically, one assumes additional properties for a risk measure $\rho$ such as {\it counter monotonicity}, {\it convexity}, and {\it translation invariance}. For details, see \cite{FolSch2002} and references therein. For brevity, as we are interested in problems related to estimation of $\rho$, we have decided not to repeat detailed definitions here. Let us alone mention that from a financial point of view the value $\rho(X)$ is a quantification of risk for financial future position $X$ and is often interpreted as the amount of money one has to add to the position $X$ such that $X$ becomes acceptable. Hence, positions $X$ with $\rho(X) \le 0$ are considered acceptable (without adding additional money).

A priori, the definition of a risk measure is formulated without any relation to the underlying probability. However, in most practical applications  one typically considers law-invariant risk-measures. Roughly spoken, a risk-measure is called {law-invariant} with respect to a probability measure $P$, if $\rho(X)=\rho(\tilde X)$ whenever the laws of $X$ and $\tilde X$ coincide under $P$, see for example \cite[Section 5]{FolKni2013}. 

Hence, $\rho$ typically depends on the underlying probability measure $P_\theta$ and consequently, we obtain a family of risk-measures $(\rho_{\theta})_{\theta\in\Theta},$ which we again denote by $\rho$. Here,  $\rho_\theta$ is the risk-measure obtained under $P_\theta$. 
Being law-invariant, the risk-measure can be identified with a function of the cumulative distribution function of $X$. More precisely, we obtain the following definition. Denote by $\cD$ the convex space of cumulative distribution functions of real-valued random variables. 
\begin{definition}
The family of risk-measures $(\rho_{\theta})_{\theta\in\Theta}$ is called \emph{law-invariant}, if there exists a function $R:\cD \to \R\cup \{+\infty\}$ such that for all $\theta \in \Theta$ and $X\in L^0$
\begin{align}\label{def:R}
\rho_\theta(X)=R(F_X(\theta)),   
\end{align}
 $F_X(\theta)=P_\theta(X \le \cdot)$ denoting the cumulative distribution function of $X$ under the parameter $\theta$.  
\end{definition}

We aim at estimating the risk of the future position  when $\theta\in\Theta$ is unknown and needs to be estimated from a data sample $x_1,\dots,x_n$. If $\theta$ were known, we could directly compute the corresponding risk measure $\rho_{\theta}$ from $P_\theta$ and would not need to consider the family $(\rho_{\theta})_{\theta\in\Theta}$.  Various estimation methodologies are at hand, the most common one being the plug-in estimation; see Section~\ref{S:Plug-in} for details.

\begin{definition}\label{def:estimator}
An \emph{estimator} of a risk measure is a Borel function $\hat{\rho}_n: \R^n\to \bR\cup\{+\infty\}$.
\end{definition}
Sometimes we will call $\hat \rho_n$ also risk estimator. The value $\hat{\rho}_{n}(x_1,\dots,x_n)$ corresponds to the \emph{estimated} amount of capital which should classify, after adding the capital to the position, the future position $X$  acceptable.

Given random sample $X_1,X_2,\dots,X_n,$ we also denote by $\hat \rho_n$  the random variable
\[
\hat{\rho}_n(\omega):=\hat{\rho}_n(X_1(\omega),\dots,X_n(\omega)), \quad \omega \in \Omega,
\]
corresponding to the estimator $\hat \rho_n$. 
By $\hat \rho$ we denote the sequence of risk estimators $\hat \rho = (\hat \rho_n)_{n\in\bN}$. If there is no ambiguity, we will call $\hat{\rho}$ also risk estimator and sometimes even write $\hat\rho$ instead of $\hat\rho_n$.

The concept of {\it estimator} given in Definition~\ref{def:estimator} is very general. One very common way in practical estimation of risk measures is to separate the estimation of the distribution of the underlying random variable from the estimation of the risk measure. This leads to the well established plug-in estimators, which we discuss in the following section.

\section{Plug-in estimation}\label{S:Plug-in}

A common way to estimate risk is the \emph{plug-in estimation}; see \cite{Ace2007,ConDegSca2010,FolKni2013} and references therein. The idea behind this approach is to use the highly developed tools for estimating the distribution function of $X$ and plug in this estimate into the desired risk measure. Denote the estimator of the unknown distribution by $\hat{F}_{X}$ and recall the function $R$ from Equation \eqref{def:R}. Then the \emph{plug-in estimator} $\hat{\rho}_{\textrm {plugin}}$ is given by 
\begin{equation}\label{eq:plugin-nonpar}
\hat{\rho}_{\textrm {plugin}}(x):=R(\hat{F}_{X}).
\end{equation}

More specifically, considering the parametric case with a family of law-invariant risk-measure $(\rho_\theta)_{\theta \in \Theta}$, we obtain from ~\eqref{def:R}, that the plug-in estimator is given by
\[
\hat{\rho}_{\textrm {plugin}}(x)=R(F_X(\hat \theta))=\rho_{\hat\theta}(X),
\]
where $\hat\theta$ denotes parameter estimator given sample $x$.

Let us now present some specific examples, where we provide explicit formulas for the plug-in estimators of the considered risk both for non-parametric and parametric case.

\begin{example}[Empirical distribution plug-in estimation]\label{ex:307}
As an example we could use the empirical distribution for the plug-in estimator. The assumption is that $X_1,\dots,X_n$ are independent, having the same distribution like $X$,  and a sample $x=(x_1,\ldots,x_n)$ is at hand. Then, the empirical distribution is given by
\[
\hat{F}_X(t):=\frac{1}{n}\sum_{i=1}^{n}\Ind_{\{x_i\leq t\}}, \qquad t \in \R,
\]
where $\Ind_{A}$ is indicator of event $A$. It is a discrete distribution and hence $R(\hat F_X)$ is easy to compute.
\end{example}

\begin{example}[Kernel density estimation]\label{ex:kernel}
Assuming that $X$ is (absolutely) continuous, one of the most popular non-parametric estimation techniques for the density of $X$ is the so-called {\it kernel density estimation}, see for example \cite{Ros1956,Par1962}. Instead of estimating the distribution itself, one focusses on estimating the probability density function, as in the continuous case we could recover one from another. Given the sample $x=(x_1,\ldots,x_n)$, kernel function $K:\bR\to\bR$ and bandwidth parameter $h>0$ (see \cite{Sil1986} for details), the estimator $\hat f$ for the unknown density $f$ is given by
\[
\hat f(z)=\frac{1}{nh} \sum_{i=1}^{n} K\left(\frac{z-x_i}{h}\right),\quad z\in\mathbb{R}.
\]
The most popular kernel functions are the Gaussian kernel $K^1$ and the Epanechnikov kernel $K^2$, given by
\[
K^1(u)=\frac{1}{\sqrt{2\pi}} e^{-\frac 12u^{2}}\quad\textrm{and}\quad K^2(u)=\frac 34(1-u^2)\,\mathbf {1} _{\{|u|\leq 1\}}.
\]
The optimal value of the bandwidth parameter could also be estimated, but this depends on additional assumptions. For example, one could show that if the sample is Gaussian, then the optimal choice of bandwidth parameter is approximately $1.06\hat\sigma n^{-1/5}$, where $\hat\sigma$ is the standard deviation of the sample.
\end{example}

\begin{example}[Plug-in estimators under normality]\label{ex:plugin}
Let us assume that $X$ is normally distributed under $P_{\theta}$, for any $\theta=(\theta_1,\theta_2)\in\Theta=\R\times\R_{>0}$, where $\theta_1$ and $\theta_2$ denote mean and variance, respectively. Given sample $x=(x_1,\ldots,x_n)$, let $\hat\theta_1$ and $\hat\theta_{2}$ denote the estimated parameters (obtained e.g. using MLE). Then, assuming that $\rho$ is {\it translation invariant} and {\it positively homogenous} (see \cite{FolSch2002} for details), the classical  \emph{plug-in estimator} $\hat{\rho}$ can be computed as follows
\begin{equation}\label{eq:normal.plugin}
\hat{\rho}(x)=R(F_X(\hat\theta))=\rho_{\hat{\theta}}(X)=\rho_{\hat\theta}\left(\hat\theta_1\, \frac{X-\hat\theta_2}{\hat\theta_1}+\hat\theta_2\right)=-\hat\theta_2+\hat\theta_1R(\Phi),
\end{equation}
where $\Phi$ denotes the cumulative distribution function of the standard normal distribution. If we are interested in estimating the value-at-risk, then the estimator given in Equation \eqref{eq:normal.plugin} coincides with the one defined in Equation \eqref{eq:var.norm}.
\end{example}

\begin{example}[Plug-in estimator for the $t$-distribution]\label{ex:student}
Assume now that $X$ has a generalized $t$-distribution under $P_{\theta}$, for any $\theta=(\theta_1,\theta_2,\theta_3)\in\Theta=\R\times\R_{>0}\times\N_{>2}$, where $\theta_1$, $\theta_2$ and $\theta_3$ denote mean, variance and degrees of freedom parameter, respectively. Given the sample $x=(x_1,\ldots,x_n)$, let $\hat\theta_1$, $\hat\theta_{2}$ and $\hat\theta_{3}$ denote the estimated parameters (obtained e.g. using Expectation-Maximization method; see \cite{FerSte1998} for details). Then, assuming that $\rho$ is {\it translation invariant} and {\it positively homogenous}, the \emph{plug-in estimator} can be expressed as
\begin{equation}\label{eq:var.stud}
\hat{\rho}(x)=-\hat\theta_1+\hat\theta_2 \sqrt{\frac{\hat\theta_3-2}{\hat\theta_3}}R(t_{\hat\theta_3}),
\end{equation}
where $t_v$ corresponds to the standard $t$-distribution with $v$ degrees of freedom. In particular, for value-at-risk at level $\alpha$ we get $R(t_{\hat\theta_3})=-t^{-1}_{\hat\theta_3}(\alpha)$.
\end{example}

\begin{example}[Plug-in estimator using extreme-value theory]\label{ex:gpd}
Let us assume that $X$ is absolutely continuous for any $\theta\in\Theta$. For any threshold level $u<0$ we define the conditional excess loss distribution of $X$ under $\theta\in\Theta$ as
\[
[F_{X}]_{u}(\theta,t)=P_{\theta}(X\leq u+t|X< u),\quad \textrm{for } t\leq 0.
\]
Roughly speaking, The Pickands--Balkema--de Haan theorem states that for any $\theta\in\Theta$, if $u\to -\infty$, then the conditional excess loss distribution should converge to some Generalized Pareto Distribution (GPD).\footnote{Under some mild condition imposed on distribution $F_{X}(\theta)$. This includes e.g. families of normal, lognormal, $\chi^2$, $t$, $F$, gamma, exponential, and uniform distributions.} We refer to~\cite{McN1999, MFE} and references therein for more details.
This result can be used to provide an approximative formula for the risk measure estimator, if the risk measure solely depends on the lower tail of $X$. This is the case e.g. for value-at-risk or expected shortfall, especially when small risk levels $\alpha\in (0,1)$ are considered. Given threshold level $u<0$, sample $x=(x_1,\ldots,x_n)$ and using so called {\it Historical-Simulation Method} (see e.g. \cite{McN1999} for details) we define $\hat{F}_X$ for any $t<u$ setting
\begin{equation}\label{eq:GPD.dist}
\hat F_X(t)=\frac{k}{n}\left(1+\hat\xi\frac{u-t}{\hat\beta}\right)^{-1/\hat\xi},
\end{equation}
where $k$ is the number of observations that lie below threshold level $u$ and $(\hat\xi,\hat\beta)$ correspond to shape and scale estimators in the GPD family.
The estimators $\hat\xi$ and $\hat\beta$ can be computed taking only (negative values of) observations that lie below threshold level $u$ and using e.g.~the {\it Probability Weighted Moments Method} (again, see \cite{McN1999} and references therein for details). Now, assuming that the function $R$ given in \eqref{def:R} depends only on the tail of the distribution, i.e. for any $\theta\in\Theta$ we only need $F_X(\theta)|_{ (-\infty,u)}$ to calculate $R(F_X(\theta))$, one could obtain the formula for the plug in estimator using \eqref{eq:GPD.dist}. In particular, for value-at-risk at level $\alpha\in (0,1)$, if only $\alpha<\hat F_X(u)$, then we obtain the estimator
 \begin{equation}\label{eq:var.gpd}
\hat{\var}_{\alpha}^{\textrm{GPD}}(x)=\hat F^{-1}_X(\alpha)=-u+\frac{\hat\beta}{\hat\xi}\left(\left(\frac{\alpha n}{k}\right)^{-\hat\xi}-1\right).
\end{equation}
Please note that this estimator might be in fact considered non-parametric, as it approximates the value of $\rho(X)$ for a large class of distributions including almost all ones used in practice.
\end{example}

\section{Unbiased estimation of risk}\label{S:Robust}

Quite surprisingly, the plug-in procedure often leads to an underestimation of risk as already explained  in Section~\ref{S:example}. It is our goal to introduce a new class of estimators, which we call unbiased, not suffering from this deficiencies.

\begin{definition}\label{def:unbiased}
The  estimator $\hat \rho_n$ is called \emph{unbiased} for $\rho(X)$, if for all $\theta \in \Theta$,
\begin{equation}\label{eq:rho.unbiased}
\rho_\theta(X+\hat \rho_n)=0.
\end{equation}
\end{definition}
An unbiased estimator has the economically desirable feature, that adding the estimated amount of risk capital $\hat \rho_n$ to the position $X$ makes the position $X+\hat \rho_n$ acceptable under all possible scenarios $\theta \in \Theta$. Requiring equality in Equation \eqref{eq:rho.unbiased} ensures that the estimated capital is not too high.  
It is worth pointing out, that except for the i.i.d.\ case, the distribution of $X+\hat \rho_n$ does also depends on the dependence structure of 
$X,X_1,\dots,X_n$ and not only  on the (marginal) laws. 

From the financial point of view, given a historical data set, or even a stress scenario $(x_1,\ldots,x_n)$, the number $\hat{\rho}_{n}(x_1,\ldots,x_n)$ is used to determine the capital reserve for position $X$, i.e. the minimal amount for which the risk of the secured position $\xi^n(x_1,\ldots,x_n):=X+\hat{\rho}_{n}(x_1,\ldots,x_n)$ is acceptable. As the parameter $\theta$ is unknown, it would be highly desirable to minimise the risk of the secured position $\xi^n$, now considered as the function of random variables $X_1$, $\ldots$, $X_n$. If we do this for any  $\theta\in\Theta$, then our estimated capital reserve would be close to the real (unknown) one. To do so, we want the (overall) risk of position $\xi^n$ to be equal to 0, for any value of $\theta\in\Theta$. This is precisely  unbiasedness in the sense of Definition~\ref{def:unbiased}.

\begin{remark}[Relation to the statistical definition of unbiasedness]
Definition \ref{def:unbiased} differs from unbiasedness in the statistical sense: the estimator $\hat \rho_n$ is called \emph{statistically unbiased}, if 
\begin{equation}\label{eq:rho.unbiased.classic}
E_{\theta}[\hat \rho_n]=\rho_{\theta}(X), \qquad \text{for all }\theta \in \Theta,
\end{equation}
where $E_{\theta}$ denotes the expectation operator under $P_{\theta}$. 
While the condition \eqref{eq:rho.unbiased.classic} is always desirable from the theoretical point of view, it might not be prioritised by the institution interested in risk measurement, as the mean value of estimated capital reserve does not determine the risk of the secured position $X+\hat{\rho}_n$. Let us explain this in more detail:
In practice, the main goal is to define an estimator in such a way, that it behaves well in various backtesting or stress-testing procedures. The types of conducted tests are usually given by the regulator (see for example  \citepalias{Bas2011}). In the case of value-at-risk the so-called {\it failure rate} procedure is often considered. As explained in Section~\ref{S:example}, this procedure focusses on the rate of exceptions, i.e. ratio of scenarios in which estimated capital reserve is insufficient. This nonlinear function is different from the linear measure given by the average value of estimated capital reserve, underlining the need for a different definition of a bias. See also Remark \ref{rem:probability.bias} for further explanation.
\end{remark}

\begin{remark}[Relation to probability unbiasedness]\label{rem:probability.bias}
In \cite{FraHer2012}, the authors introduced the concept of a \emph{probability unbiased} estimation:  denote by $F_X(\theta,t)=P_\theta(X \le t)$, $t \in \R$  the cumulative distribution function of $X$ under $P_\theta$. Then the estimator $\hat \rho_n$ is called \emph{probability unbiased}, if
\begin{equation}\label{eq:prob.bias}
E_{\theta}[F_{X}(\theta,-\hat{\rho}_n)]=F_{X}(\theta,-\rho_{\theta}(X)),\qquad \textrm{ for all } \theta\in\Theta.
\end{equation}
Intuitively, the left hand side corresponds to the average probability, that our estimated capital reserve would be insufficient, while the right hand side corresponds to the probability of insufficiency of the theoretical capital reserve.
This approach is proper for value-at-risk in the strongly restricted setting of the i.i.d.\ Example \ref{ex:iid}. In fact, in that setting, it coincides with our definition of  unbiasedness from Definition \ref{def:unbiased}: 
indeed, assume that $F_X(\theta)$ is continuous and that $X_1,\dots,X_n,X$ are i.i.d. Then  $\hat\rho_n$ and $X$ are independent and hence
\[
E_{\theta}[F_{X}(\theta,-\hat{\rho}_n)]=P_{\theta}[X+\hat\rho_n<0].
\]
On the other hand we know that, for $\rho_{\theta}$ being value-at-risk at level $\alpha$, we obtain $F_{X}(\theta,-\rho_{\theta}(X))=\alpha$, so \eqref{eq:prob.bias} is equivalent to
\begin{equation}\label{eq:var.bt}
P_{\theta}[X+\hat\rho_n<0]=\alpha.
\end{equation}
Now it is easy to show that this is equivalent to  
\[
\rho_{\theta}(X+\hat\rho_n)=\inf\{x\in\bR\colon P_{\theta}[X+\hat\rho_n+x<0]\leq \alpha\}=0.
\]
In the general case we consider here, a more flexible concept is needed to define the risk estimator bias. In particular, the average probability of insufficiency does not contain information about the level of capital deficiency. This, however, is a key concept, e.g., when considering expected shortfall; compare Example \ref{ex:3}.
\end{remark}


\begin{remark}[Relation to level adjustment]\label{rem:level_adjustment}
 In \cite{Frank2016} and \cite{FraHer2012}, an adjustment of the level $\alpha$ has been proposed to take the bias into account. The methodology is tailor-made to the unbiased estimation of the probability level of crossings (exceptions); see Remark~\ref{rem:probability.bias}. We discuss this issue using the notation introduced in Remark~\ref{eq:prob.bias}. The  value-at-risk at level $\alpha$ if the parameter $\theta \in \Theta$ were known is given by $-F_X^{-1}(\theta,\alpha)$. Then the expected number of exceedances of the position $X$ obtained by adding the value-at-risk to the position equals $\alpha$:
 \[
E_{\theta}[ \ind {X- F_X^{-1}(\theta,\alpha) < 0}] = P_{\theta}[X<F_X^{-1}(\theta,\alpha)] = \alpha.
\]
 
 In fact, estimation can not only be done for a single $\alpha$, but for all $\alpha\in(0,1)$. We denote the estimators by $\hat \rho(\alpha)$ and observe that, as function of $\alpha$, they are typically continuous and decreasing (the lower $\alpha$, the higher risk capital is needed to ensure that the probability crossing this levels is as small as $\alpha$). If an established estimation procedure is at hand, represented by the family $(\hat \rho(\alpha))_{\alpha \in (0,1)}$, one can  adjust $\alpha$ to remove  estimation bias. In particular, if the average exceedance rate after estimation should match $\alpha$ one will look for the adjusted level $\alpha_{\textrm{adj}}$, such that
\[
\alpha = E_{\theta}[ \ind {X- \hat \rho(\alpha_{\textrm{adj}})< 0}] = \int F_X(\theta,y) p_{\hat \rho(\alpha_{\textrm{adj}})} (dy);
\]
 here, $p_{\hat \rho(\alpha_{\textrm{adj}})}$ denotes the density of $\hat \rho(\alpha_{\textrm{adj}})$. The required $\alpha_{\textrm{adj}}$ can be found numerically when the density of the estimator is at hand. Note that this requirement exactly matches \eqref{eq:var.bt} with the difference that not the estimator is modified but the level $\alpha$ adjusted. A comparison to the unbiased estimator in Equation \eqref{eq:var.bia} reveals that the adjustment  $\alpha_{\textrm{adj}}$ in the Gaussian case also depends on the sample size $n$, which might be undesirable in practice. Also, this method is specific to value-at-risk. We refer to \cite{FraHer2012} and \cite{Frank2016} for details and examplary level adjustment algorithms.
 \end{remark}

\begin{remark}[Relation to subadditive and loss-based risk measures]
It is interesting to analyze the minimal requirements which render unbiasedness as in Definition \ref{def:unbiased} useful: the only requirement is that a position $X\in L^0$ is acceptable, if $\rho(X)\le 0$. This is directly linked to a proper normalization of $\rho$. Even more, it does not require that $\rho$ is coherent, nor even that it is loss-based, as in \cite{ContDeguestHe}, or subadditive, as in \cite{ElKarouiRavanelli}. Consequently, the proposed estimation methodology also applies to these interesting classes of risk measures. A detailed analysis is, however, beyond the scope of this article. 
\end{remark}

\section{Examples}\label{S:examples}

In this section we precent some examples highlighting the application of the concept of unbiased risk estimators.

\begin{example}[Unbiased estimation of the mean]\label{ex:1}
Assume that $X$ is integrable for any $\theta\in\Theta$, and consider a position acceptable if it has non-negative mean. This corresponds to the  family $\rho$ of risk measures 
\[
\rho_{\theta}(X)=E_{\theta}[-X],\qquad \theta\in\Theta.
\]
Clearly $\rho$ is law-invariant. Corresponding to Equation \eqref{eq:rho.unbiased}, a risk estimator $\hat \rho$ is unbiased in this setting if 
\[
0=\rho_\theta(X+\hat \rho)=E_{\theta}[-(X+\hat{\rho})]=\rho_{\theta}(X)-E_\theta[\hat{\rho}].
\]
Therefore, the estimator $\hat \rho$ is unbiased if and only if it is statistically unbiased for any $\theta\in\Theta$. Hence, the negative of the sample mean, given by
\[
\hat{\rho}_{n}(x_1,\ldots,x_n)=-\frac{\sum_{i=1}^{n}x_i}{n}, \qquad n \in \bN,
\]
is an unbiased estimator of the risk measure of position $X$.
\end{example}

\begin{example}[Unbiased estimation of value-at-risk under normality]\label{ex:2}
Let $X$ be normally distributed with mean $\theta_1$ and variance $\theta_2$ under $P_{\theta}$, for any $\theta=(\theta_1,\theta_1)\in\Theta=\R\times\R_{>0}$. For a fixed $\alpha\in (0,1)$, let
\begin{align}\label{def:var}
\rho_{\theta}(X)=\inf\{x\in\bR\colon P_{\theta}[X+x<0]\leq \alpha\},\qquad \theta\in\Theta,
\end{align}
denote value-at-risk at level $\alpha$. As $X$ is absolutely continuous, unbiasedness as defined in Equation \eqref{eq:rho.unbiased} is equivalent to
\begin{equation}\label{eq:var.2}
P_{\theta}[X+\hat{\rho}<0]=\alpha, \qquad \text{for all }\theta \in \Theta.
\end{equation}
This concept coincides with the definition of a \emph{probability unbiased} estimator of value-at-risk (see Remark~\ref{rem:probability.bias} for details). We define estimator $\hat{\rho}$, as 
\begin{equation}\label{eq:est.var}
\hat{\rho}(x_1,\ldots,x_n)=-\bar{x}-\bar{\sigma}(x)\sqrt{\frac{n+1}{n}}t^{-1}_{n-1}(\alpha),
\end{equation}
where $t_{n-1}$ stands for cumulative distribution function of the student-$t$ distribution with $n-1$ degrees of freedom and
\[
\bar{x}:=\frac{1}{n}\sum_{i=1}^{n}x_{i},\quad \bar{\sigma}(x):=\sqrt{\frac{1}{n-1}\sum_{i=1}^{n}(x_i-\bar{x})^{2}},
\]
denote the efficient estimators of mean and standard deviation, respectively. We show that the estimator $\hat{\rho}$ is an unbiased risk estimator:  note that $X\sim \cN(\theta_1,(\theta_2)^2)$ under $P_\theta$. Using the fact that $X$, $\bar{X}$ and $\bar{\sigma}(X)$ are independent for any $\theta\in\Theta$ (see, e.g., \cite{Bas1955}), we obtain 
\[
T:=\sqrt{\frac{n}{n+1}}\cdot\frac{X-\bar{X}}{\bar{\sigma}(X)}=\frac{X-\bar{X}}{\sqrt{\frac{n+1}{n}}\theta_2}\cdot\sqrt{\frac{n-1}{\sum_{i=1}^{n}(\frac{X_i-\bar{X}}{\theta_2})^{2}}}\sim t_{n-1}.
\]
Thus, the random variable $T$ is a pivotal quantity and
\[
P_{\theta}[X+\hat{\rho}<0]=P_{\theta}[T<q_{t_{n-1}}(\alpha)]=\alpha,
\]
which concludes the proof.
\end{example}

\begin{remark}\label{rem:gaussian.bias}
It follows that the difference between Gaussian unbiased estimator defined in \eqref{eq:est.var} and the classical  plug-in Gaussian estimator given in \eqref{eq:var.norm} is equal to
\begin{align}
\hat{\var}_{\alpha}^{\textrm{u}}(x)-\hat{\var}_{\alpha}^{\textrm{norm}}(x)=-\bar{\sigma}(x)\left(\sqrt{\frac{n+1}{n}}t^{-1}_{n-1}(\alpha)-\Phi^{-1}(\alpha)\right).
\end{align}
Consequently, as $\bar\sigma(x)$ is consistent, and
\[
\sqrt{\frac{n+1}{n}}t^{-1}_{n-1}(\alpha)\xrightarrow{n\to\infty} \Phi^{-1}(\alpha),
\]
we obtain that the bigger the sample, the closer the estimators are to each other -- the bias of plug-in estimator decreases.
\end{remark}

The procedure from the previous example can be applied to almost any (reasonable) coherent risk measure. We choose expected shortfall as an example to illustrate how this can be achieved.
\begin{example}[Unbiased estimation of expected shortfall under normality]\label{ex:3}
As before, let $X$ be normally distributed with mean $\theta_1$ and variance $\theta_2$ under $P_{\theta}$, for any $\theta=(\theta_1,\theta_1)\in\Theta=\R\times\R_{>0}$. Let us fix $\alpha\in (0,1)$. The expected shortfall at level $\alpha$ under a continuous distribution is given by
\[
\rho_{\theta}(X)=E_{\theta}[-X | X\leq q_{X}(\theta,\alpha)],
\]
where $q_{X}(\theta,\alpha)$ is $\alpha$-quantile of $X$ under $P_{\theta}$, that coincides with the negative of value-at-risk at level $\alpha$ from Equation \eqref{def:var}; see Lemma 2.16 in \cite{MFE}. Due to translation invariance and positive homogeneity of $\rho_{\theta}$, exploiting the fact that $X$, $\bar{X}$ and $\bar{\sigma}(X)$ are independent for normally distributed $X$, a good candidate for $\hat{\rho}$ is
\begin{equation}\label{eq:est.tvar}
\hat{\rho}(x_1,\ldots,x_n)=-\bar{x}-\bar{\sigma}(x)a_n,
\end{equation}
for some $(a_n)_{n\in\bN}$, where $a_n\in\bR$. There exists a sequence $(a_n)_{n\in\bN}$ such that $\hat \rho$ is unbiased:  As $\rho_{\theta}$ is positively homogeneous, we obtain for all $\theta \in \Theta$
\begin{align}
\rho_{\theta}(X+\hat{\rho}) &=\theta_2  \sqrt{\frac{n+1}{n}}\, \rho_{\theta}\bigg(\frac{X-\bar{X}-a_n\bar{\sigma}(X)}{\theta_2\sqrt{\frac{n+1}{n}}}\bigg)\nonumber\\
& =\theta_2  \sqrt{\frac{n+1}{n}}\, \rho_{\theta}\Bigg(\frac{X-\bar{X}}{\theta_2\sqrt{\frac{n+1}{n}}}-\frac{a_n\sqrt{n}}{\sqrt{(n-1)(n+1)}}\cdot \sqrt{n-1}\frac{\bar{\sigma}(X)}{\theta_2}\Bigg)\nonumber\\
& =\theta_2  \sqrt{\frac{n+1}{n}}\,\rho_{\theta}\bigg(Z-\frac{a_n\sqrt{n}}{\sqrt{(n-1)(n+1)}}V_n\bigg),\label{eq:tvar.an2}
\end{align}
where, $Z\sim\cN(0,1)$, $V_n\sim \chi_{n-1}$ and both being independent. Note that the distribution of $(Z,V_n)$ does not depend on $\theta$. Thus, it is enough to show that there exists $b_n\in\bR$ such that
\begin{equation}\label{eq:est.tvar.cond}
\rho_{\theta}\left(Z+b_nV_n\right)=0.
\end{equation}
As $V_{n}$ is non-negative and the risk measure $\rho_{\theta}$ is counter-monotone, we obtain that \eqref{eq:est.tvar.cond} is decreasing with respect to $b_n$. Moreover, $0<\rho_{\theta}(Z)=\rho_{\theta}(Z+0V_n)$. For $b_n$ large enough we get $\rho_{\theta}(Z+b_nV_n)<0$, as $\rho_{\theta}(Z+b_nV_n)=b_n\rho_{\theta}(\frac{Z}{b_n}+V_n)$ and
\[
\rho_{\theta}\Big(\frac{Z}{b_n}+V_n\Big)\xrightarrow{b_n\to\infty} \rho_{\theta}(V_n)< 0,
\]
due to the Lebesgue continuity property of expected shortfall on $L^{1}$ (see \cite[Theorem 4.1]{KaiRus2009}). Thus, again using continuity of $\rho_{\theta}$, we conclude that there exists $b_n\in\bR$ such that \eqref{eq:est.tvar.cond} holds. Moreover, the value of $b_n$ is independent of $\theta$, as the family $(\rho_{\theta})_{\theta\in\Theta}$ is law-invariant (see Equation \eqref{def:R}) and $Z$, $V_{n}$ are  pivotal quantities. Note that we only needed positive homogeneity and monotonicity of $\rho_{\theta}$ as well as \eqref{eq:est.tvar.cond} to show the existence of an unbiased estimator. Moreover, the value of $b_n$ in \eqref{eq:est.tvar.cond}, and consequently $a_n$ in \eqref{eq:est.tvar}, can be computed numerically without effort.
\end{example}


\section{Asymptotically  unbiased estimators}
\label{sec:asymptotics}

Even if the risk estimators from Examples \ref{ex:307}, \ref{ex:kernel}, \ref{ex:plugin}, \ref{ex:student} and \ref{ex:gpd} are biased (cf. Table~\ref{t:test1}), one might still have nice properties in an asymptotic sense which we study in the following. 

\begin{definition}
A sequence of risk estimators $\hat\rho = (\hat \rho_n)_{n \in \N}$ will be called \emph{unbiased} at $n\in\N$, if $\hat \rho_n$ is unbiased. If unbiasedness holds for all $n\in\bN$, we call the sequence $\hat \rho$ unbiased. The sequence $\hat{\rho}$ is called  \emph{asymptotically unbiased}, if 
\[
\rho_\theta(X+\hat{\rho}_{n})\xrightarrow{n\to\infty} 0, \qquad \text{ for all }\theta \in \Theta.
\]
\end{definition}

In many cases the estimators of the distribution are consistent in the sense that $\hat F_X \to F_X(\theta)$. Indeed, the Glivenko-Cantelli theorem gives even uniform convergence over convex sets of the empirical distribution with probability one. Intuitively, if the underlying distribution-based risk measure admit some sort of continuity, then we could expect that the the plug-in estimator satisfies
\[
\hat{\rho}_{n}\xrightarrow{n\to\infty} \rho_{\theta}(X)
\]
almost surely for each $\theta \in \Theta$. 
Consequently, for any $\theta\in\Theta$ we also would get
\[
\rho_{\theta}(X+\hat{\rho}_{n})\xrightarrow{n\to\infty}  \rho_{\theta}(X+\rho_{\theta}(X))=0,
\]
which is exactly the definition of asymptotic  unbiasedness. Let us now present two examples, which show asymptotic unbiasedness of the empirical value-at-risk estimator \eqref{eq:var.hist} and the plug-in Gaussian estimator for expected shortfall.  

\begin{remark}
The proposed definition of asymptotical unbiasedness has similarities to the notion of consistency suggested in \cite{Davis2016}. This notion of consistency requires that averages of the calibration errors converge suitable fast to $0$ when the time period tends to infinity. 
Hence, asymptotically unbiased risk estimators will be consistent when the calibration error is measured with the risk measure itself. On the other side, it should be noted that our main goal is to obtain the optimal risk estimator without averaging out under- or overestimates as they have an asymmetric effect on the portfolio performance. 
\end{remark}

We obtain the following result. Recall that we study an i.i.d.\ sequence $X,X_1,X_2,\dots$ 
Let $\alpha\in (0,1)$ and consider the negative of emprical $\alpha$-quantile
\begin{align}\label{def:empvar}
\hat{\rho}_{n}(x_1,\ldots,x_n)=-x_{(\lfloor n\alpha\rfloor+1)}, \quad n \in \N,
\end{align}
 which we call empirical estimator of value-at-risk at level $\alpha$ (compare also \eqref{eq:var.hist}). By $\hat \rho_n$ we denote the random variable $\hat \rho_n(X_1,\dots,X_n)$.
\begin{proposition}\label{prop:1}
Assume that  $X$ is absolutely continuous under $P_{\theta}$ for any $\theta \in \Theta$. The sequence of empirical estimators of value-at-risk  given in \eqref{def:empvar} is asymptotically unbiased.
\end{proposition}

\begin{proof}
The proof directly follows from asymptotical properties of empirical quantiles. For the readers' convenience we provide a sketch of the proof.  In this regard, for any $\epsilon>0$ and $\theta\in\Theta$, let $A_{n,\epsilon}:=\left\{\, |\hat{\rho}_{n}+F_{X}^{-1}(\theta,\alpha)|\geq\epsilon\right\}$, where $F_{X}^{-1}(\theta,\cdot)$ denotes the inverse of $F_X(\theta,\cdot)$. Then we have that
\begin{align*}
P_{\theta}[X+\hat{\rho}_{n}<0] &\geq P_{\theta}[A^{c}_{n,\epsilon}]F_{X}\big(\theta,F^{-1}_{X}(\theta,\alpha)-\epsilon\big)-P_{\theta}[A_{n,\epsilon}],\\
P_{\theta}[X+\hat{\rho}_{n}<0] &\leq P_{\theta}[A^{c}_{n,\epsilon}]F_{X}\big(\theta,F^{-1}_{X}(\theta,\alpha)+\epsilon\big)+P_{\theta}[A_{n,\epsilon}].
\end{align*}
Using the fact that empirical value-at-risk estimator is consistent \cite[Example 2.10]{ConDegSca2010}, i.e. for any $\theta\in\Theta$, under $P_{\theta}$, we get
\[
\hat{\rho}_{n}\xrightarrow{n\to\infty} -F_{X}^{-1}(\theta,\alpha) \qquad a.s.,
\]
we obtain that $P_{\theta}[A_{n,\epsilon}]\xrightarrow{n\to\infty} 0$, for any $\epsilon>0$ and $\theta\in\Theta$. Consequently,
\[
F_{X}\big(\theta,F^{-1}_{X}(\theta,\alpha)-\epsilon\big) \leq \lim_{n\to\infty}P_{\theta}[X+\hat{\rho}_{n}<0]\leq F_{X}\big(\theta,F^{-1}_{X}(\theta,\alpha)+\epsilon\big),
\]
for any $\epsilon>0$ and $\theta\in\Theta$. Taking the limit, and noting that $F_{X}(\theta,\cdot)$ is continuous, we get
\[
P_{\theta}[X+\hat{\rho}_{n}<0]\xrightarrow{n\to\infty} F_{X}\big(\theta,F^{-1}_{X}(\theta,\alpha)\big)=\alpha,
\]
for any $\theta\in\Theta$, which concludes the proof, due to \eqref{eq:var.2}.
\end{proof}

In a similar fashion, one can prove that estimator given in~\eqref{eq:var.hist} is asymptotically unbiased as well. Moreover, slightly changing the proof of Proposition~\ref{prop:1} one could show that under normality assumption the sequence of classical plug-in Gaussian estimators of value-at-risk given in \eqref{eq:var.norm} is also asymptotically unbiased. See also Remark~\ref{rem:gaussian.bias}. 


In a similar way we obtain asymptotic unbiasedness of the Gaussian plug-in expected shortfall estimator introduced in \eqref{eq:normal.plugin}:
In this regard let $X$ be normally distributed with mean $\theta_1$ and standard deviation $\theta_2$ under $P_{\theta}$, for any $(\theta_1,\theta_2)=\theta\in\Theta=\R\times\R_{>0}$. For a fixed $\alpha\in (0,1)$, let
\[
\rho_{\theta}(X)=E_{\theta}[-X | X\leq q_{X}(\theta,\alpha)],
\]
denote the expected shortfall at level $\alpha$. Following \eqref{eq:normal.plugin}, set
\begin{align}\label{eq:392}
\hat{\rho}_n(x_1,\ldots,x_n)=-\bar{x}+\bar{\sigma}(x)R(\Phi), \quad n \in \N,
\end{align}
where $\Phi$ is a Gaussian distribution and  $R(\Phi)$ is the expected shortfall at level $\alpha$ under $\Phi$. The estimator \eqref{eq:392} corresponds to a standard MLE plug-in estimator, under the assumption that $X$ is normally distributed 
\begin{proposition}\label{prop:cvar}
Assume that  $X,X_1,X_2,\dots$ are i.i.d. $\cN(\theta_1,\theta_2^2)$  for any $\theta \in \Theta$. The sequence of estimators of expected shortfall given in \eqref{eq:392} is asymptotically unbiased.
\end{proposition}
\begin{proof}
First, Theorem 4.1 in \cite{KaiRus2009} shows that the tail-value-at-risk is Lebesgue-continuous, which means that for a sequence $Y_n$ converging to $Y$ almost surely and such that all $Y_n$ are dominated by a random variable being an element of $L^p$, $\lim_{n \to \infty} \rho(Y_n)=\rho(Y). $ Set $Y_n := -\bar{x}+\bar{\sigma}(x)R(\Phi)$, such that $Y_n \to \theta_1+\theta_2 R(\Phi)=:Y$ almost surely as $n\to\infty$. But, it follows directly for the tail-value-at risk under a normal distribution, denoted by $\rho_\theta$, that
$$ \rho_\theta( -\theta_1 + \theta_2 R(\Phi)) =0, $$
hence the claim. 
\end{proof}

\section{Estimating the bias and the relation to regulatory backtesting}\label{sec:backtesting}
An important concept in risk management is backtesting. Basically, backtesting procedures empirically asses the fit of the model to data measured in quantities relating to the underyling risk. Before we provide a brief comment on the relation between bias and regulatory backtesting (the detailed analysis of backtesting framework is beyond the scope of this paper) we introduce a simple procedure that could be used to measure the estimator bias for i.i.d.data.

Let us assume that we have a sample $(x_i)_{i=1,\ldots,I}$ at hand and for each element we are given the value of the estimated capital reserve denoted by $\hat\rho_i$\,.\footnote{For daily time-series analysis, the value $\hat\rho_i$ could be obtained using a simple rolling-window procedure, i.e. assuming we are also given observations before day $i$, for any given day we use past $n$ days to estimate the risk.} Then, the sample $(y_i)$ given by
\begin{equation}\label{eq:bias.ts}
y_i =x_i+\hat\rho_i, \quad i=1,\ldots,I
\end{equation}
represents secured positions and represents $X+\hat\rho$. A natural suggestion for measuring the bias of the position is to replace $\rho_\theta$ in Definition \ref{def:unbiased} by its empirical counterpart $\hat\rho_{\textrm{emp}}$. If $I$ is big enough, the  empirical estimator is expected to produce reliable results. In this light, the measure 
\begin{equation}\label{eq:bias.est}
\hat Z:=\hat\rho_{\textrm{emp}}(y_1,\ldots,y_I)
\end{equation}
is a possible quantity to assess the bias of the risk estimator. If the estimator is unbiased $\hat Z$ will be close to zero. Underestimation of risk is in turn  reflected by positive values of $\hat Z$, highlighting that the position $\hat Z$ needs additional capital to become acceptable.

Assuming that $I=250$ and the risk level is equal to $2\%$,  the empirical $\var_{0.02}$ estimator (plugged in~\eqref{eq:bias.est}) would simply compute the (negative value of the) 5-th worst outcome of the sample $(y_i)$. As an alternative, one might calculate the number of observations smaller than zero and check if their size is acceptable (i.e. close to $5$). This will relate to the standard exceedance rate test; cf. Example~\ref{ex:backt-VaR}. On the other hand, the empirical $\textrm{ES}_{0.02}$ estimator would compute the mean value for the five worst-case observations; here we can also check how many worst-case observations are needed for their mean to be positive.

As we will illustrate in the following example, at least for the value-at-risk, there is a tight connection between the regulatory backtesting framework and unbiasedness.

\begin{example}[Backtesting Value-at-Risk]\label{ex:backt-VaR}
The standard (regulatory) backtesting framework for $\var$ is based on the exceedance rate procedure; see e.g. \cite{GioLau2003} and \citepalias{Bas1996}. More precisely, one compares the average rate of exceedances over the estimated $\var$ to the expected exceedance rate, $\alpha$. The link to unbiasedness arises as follows: in the i.i.d.~case, the exception rate of $X+\hat\rho(X_1,\ldots,X_n)$ converges to the probability of the scenario in which the secured position is negative, given by
\[
\bP_{\theta}[X+\hat\rho(X_1,\ldots,X_n)<0],
\]
where $\theta$ is the unknown true parameter. On the other hand, in Remark~\ref{rem:probability.bias} we have shown that estimator $\hat\rho$ is unbiased if and only if Equality \eqref{eq:var.bt} is true for any $\theta\in\Theta$, i.e.
\[
\bP_{\theta}[X+\hat\rho_n<0]=\alpha.
\] 
Choosing the estimator in such away that the exceedance rate is close to the level $\alpha$ (done via the backtesting procedure) ensures therefore that the estimation procedure is unbiased, at least in an asymptotic sense, compare also Remark \ref{rem:level_adjustment}. 
\end{example}

\begin{remark}[On conservativeness of risk estimation]\label{rem:reg.conservative}
The regulatory $\var$ backtest classifies an estimation procedure as not appropriate if the exception-rate is {higher} than a pre-specified threshold; see \citepalias{Bas1996}. 
Consequently, such tests focus on model \emph{conservativeness} instead of model fit. In our context, rather than requiring equality in~\eqref{eq:rho.unbiased} one is interested in the property
\[
\rho_\theta(X+\hat \rho_n)\leq0,\quad \textrm{for all }\theta\in\Theta.
\]
Of course, in addition the estimation procedure is thoroughly analysed by the regulator and both, an acceptable analysis together with passing the backtest are necessary. Passing the backtest is therefore an important feature from a practical viewpoint. However, conclusions about the (overall) performance of estimators solely based on acceptable exceedance rates have to be taken with care; see the following section for additional details.
\end{remark}}

\subsection{Consistent backtesting and elicitability}\label{sec:elicitabilty}
The remarkable article \cite{Gneiting2011} critically reviews the evaluation of point forecasts. He points out that a good performance in backtesting might not necessarily imply that a given estimator is good.

\begin{example}[Perfect backtesting performance]\label{ex:perfectbacktest}
A simple, but illustrative example on the weaknesses of the exceeding rate as measurement of the quality of an estimator is as follows\footnote{This example was suggested in \cite{Holzmann2014}.}. Consider as above $I=250$ and assume that we know that the sample $(x_i)_{i=1,\dots,I}$ is centred and has support $[-1,1]$. Then, choosing 245 times the value $1$ and five times the value $-1$ gives a perfect backtesting performance when measured only by the exceedance rate. A more elaborate example is discussed in Section 1.2 of \cite{Gneiting2011}, which highlights that the estimated target needs to be related to the performance measure, which leads to the concept of \emph{elicitability}.  
\end{example}

The concept of elicitability itself origins in \cite{osband1985information}, who consider the case where a principal is contracting with a firm having superior information on future gains. The contract involves an elicitation procedure in which the firm reports cost estimates which are verified ex post. The goal of the approach is to provide a methodology which ensures truthful reporting, see also \cite{Davis2016} for an interesting discussion.

For a formal definition we follow \cite{Gneiting2011}, while introducing elicitability directly in the setting of law-invariant risk measures specified in Section \ref{S:EstimationRisk}. Recall that we consider a family of distributions $F_X(\theta),\ \theta \in \Theta$ and that a law-invariant risk measure is a function $R: \cD \to \R \cup \{+\infty\}$ mapping cumulative distribution functions to real numbers (or $+\infty$).

A \emph{scoring function} is simply a mapping $S:(\R\cup +\infty)^2 \to \R_{\ge 0}$ which compares two values of risk measures: $S(x,y)$ measures the deviation from the forecast $x$ to the realization $y$; the squared error $S(x,y)=(x-y)^2$ being a standard example.
The scoring function $S$ is called \emph{consistent} for $R$ relative to the class $\{F_X(\theta):\theta \in \Theta\},$ if 
\begin{align} \label{eq:consistent}
	E_\theta [ S(R(F_X(\theta)),Y)] \le E_\theta[ S(r,Y)]
\end{align}
for all $\theta \in \Theta$ and all $r \in \R\cup+\infty$; here $E_\theta$ denotes the expectation under which the random variable $Y$ has distribution $F_X(\theta)$. The scoring function is called \emph{strictly consistent} if it is consistent and equality in \eqref{eq:consistent} implies that $r=R(F_X(\theta))$. For example, the squared error is strictly consistent relative to the class of probability measures of finite second moment.

The risk measure $R$ is called \emph{elicitable} relative to $\{F_X(\theta):\theta \in \Theta\},$ if there exists a scoring function that is strictly consistent.
The prime example in our context is  $\var_\alpha$ (Value-at-Risk at level $\alpha$), which is elicitable with respect to the class of probability measures with finite first moment. A possible specification of a scoring function is given by
  \begin{equation}\label{eq:scor.cons.var}
 S(x,y)=(\1_{\{x\geq y\}}-\alpha)(x-y),
 \end{equation}
see Section 3.3. in \cite{Gneiting2011}. Evaluating with the performance criterion 
\begin{align}\label{eq:barS} \bar S = \frac 1 n \sum_{i=1}^n S(x_i,y_i),
\end{align}
denoting by $x_1,\dots,x_n$ the forecasts and by $y_1,\dots,y_n$ the verifying observations, guarantees that the optimal point forecast, which is the dual of the consistency property, outperforms all other estimators. This in turn  allows to identify flawed estimators like in Example \ref{ex:perfectbacktest}: indeed, in contrast to the exceedance rate, \eqref{eq:barS} involves also the distance from the estimator $x$ to the realization $y$ such that the difficulties pointed out in the example are solved by this test.

\begin{remark}[Application to backtesting]
In the context of the following empirical study (see Section ~\ref{sec:nonstandardbacktesting}), the role of the forecasts $x$ will be taken by the considered risk estimators (i.e. $-\hat\rho$) and $y$ will be the realized cash-flows.  Then, the scoring function from Equation \eqref{eq:scor.cons.var} equals 
\begin{equation}\label{eq:weighted.score}
S(-\hat\rho,X)=\alpha(X+\hat\rho)^+ + (1-\alpha)(X+\hat\rho)^-,
\end{equation}
where $\xi^+$ and $\xi^-$ denote the positive and negative part of a generic $\xi$, respectively.
 One can see that this procedure corresponds to a  weighted penalty scheme: if the secured position is positive the weight $\alpha$ is applied, while for negative secured positions the weight $(1-\alpha)$ is used. Even if being motivated by the above reasoning, this approach has  no direct link to current schemes of regulatory backtesting. One should also note that this procedure penalises estimators which are over-conservative; see Remark~\ref{rem:reg.conservative}.
\end{remark}

For the expected shortfall, the situation is more complex, as expected shortfall is itself not elicitable. However, it is jointly elicitable with value-at-risk, as recently shown in in~\cite{Fissler2016} pointing towards appropriate backtesting procedures. 
We apply these methodologies
in Section \ref{sec:nonstandardbacktesting} to our unbiased estimation procedure. 

\section{Empirical study}\label{sec:empirics}
It is the aim of this section to analyse the performance of selected estimators on various sets of real market data (Market) as well as on simulated data (Simulated). We also want to check if the statement made in Section~\ref{sec:backtesting} (about connections between unbiasedness and backtesting) is supported by the numerical analysis. Our focus is on the practically most relevant risk measures,  $\var$ and ES.

The  market data we use are returns from the data library \cite{FamFre2015}, containing returns of 25 portfolios formed on book-to-market and operating profitability in the period from 27.01.2005 to 01.01.2015. We obtain exactly $2500$ observations for each portfolio.
The sample is split into 50 separate subsets, each consisting of  50 consecutive trading days. For $k=1,2,\ldots,49$, we estimate the risk measure using the $k$-th subset and test it's adequacy on $(k+1)$-th subset; see Sections~\ref{sec5.2} and \ref{sec5.3} for  details.

While the data sample could have possible dependence and heavy tails, we consider in addition a simulation where this is not the case. The simulation allows to clearly quantify the improvement due to the bias correction. It uses  i.i.d.\ normally distributed random variables  whose mean and variance was fitted to each of the 25 portfolios. The sample size was set to 2500 for each set of parameters. 



\subsection{Regulatory backtesting Value-at-Risk}\label{sec5.2}
We begin with an evaluation of the estimators using a backtesting procedure called 
 \emph{exception-rate backtesting}, which is currently standard in the industry. However, as already detailed in Section \ref{sec:elicitabilty}, these results need to be substantiated with further evidence, which we provide in the following Section \ref{sec:nonstandardbacktesting}.

For the value-at-risk we performed the analysis on the  unbiased estimator  $\hat{\var}_{\alpha}^{\textrm{u}}$, 
the empirical sample quantile $\hat{\var}_{\alpha}^{\textrm{emp}}$, the  modified Cornish-Fisher estimator $\hat{\var}_{\alpha}^{\textrm{CF}}$ and the classical Gaussian estimator $\hat{\var}_{\alpha}^{\textrm{norm}}$   
defined in Equations \eqref{eq:var.bia}--\eqref{eq:var.norm} as well as for the GPD plug-in estimator $\hat{\var}_{\alpha}^{\textrm{GPD}}$ given in Equation~\eqref{eq:var.gpd}\footnote{For each portfolio, we set the threshold value $u$ to match the $0.7$-empirical quantile of the corresponding sample.}. As the results for all 25 portfolios were very similar, we present the results for the first five  portfolios.

The results for the exceedance rate test are reported in Table~\ref{t:boot0}, both for  market data and for the simulated data. 
They show that the unbiased estimator has in most cases a lower rate of exceedances. In particular, in the simulated Gaussian data, the  exceedance rate of the biased estimators is significantly higher than the expected rate of $0.05$ while the unbiased estimators performs very well. 

This supports the claim made in Section~\ref{sec:backtesting}. Moreover, using notation from Example~\ref{ex:2}, for the standard Gaussian estimator and under the normality assumption we know that 
\[
\bP_{\theta}[X+\hat{\var}_{\alpha}^{\textrm{emp}}<0]=\bP_{\theta}[T<\sqrt{\tfrac{n}{n+1}}\Phi^{-1}(\alpha)] >\alpha,
\]
for any $\theta\in \Theta$. Consequently, for the standard Gaussian estimator, the exception rate test should systematically produce values bigger than $\alpha$ which is indeed the case; see~Table~\ref{t:boot0}. 
Estimating the bias as laid out in Section \ref{sec:backtesting} confirms this finding and the results are not reported here. 

\begin{table}[ht]
\caption{On top, we show the results for the first five out of  25 portfolios formed on book-to-market and operating profitability in the period from 27.01.2005 to 01.01.2015 from the Fama \& French dataset, see \cite{FamFre2015} (Market). Below we show the results on simulated Gaussian data (Simulated) with mean and variance fitted to the Fama \& French portfolios.  The results for the remaining 20 portfolios show a similar behaviour and are available on request. We perform the standard backtest, splitting the sample into intervals of length 50. The table presents the average rate of exception for Value-at-Risk at level 5\% for  the  unbiased estimator  $\hat{\var}_{\alpha}^{\textrm{u}}$, 
the empirical sample quantile $\hat{\var}_{\alpha}^{\textrm{emp}}$, the  modified Cornish-Fisher estimator $\hat{\var}_{\alpha}^{\textrm{CF}}$, the GPD estimator $\hat{\var}_{\alpha}^{\textrm{GPD}}$ and the classical Gaussian estimator $\hat{\var}_{\alpha}^{\textrm{norm}}$ .   The column labelled mean shows  the mean  of all numbers in a given column, where we would expect $0.05$ if the estimator performs correctly. For the Gaussian data,  the average rate of the biased estimators is significantly higher than the expected rate of $0.05$ while the unbiased estimators perform very well (in bold type).}
\centering
\begin{tabular}{lrrrrr}\toprule
\multicolumn{1}{p{2.3cm}}{Type of data:}  & \multicolumn{5}{c}{MARKET} \\ \midrule
Portfolio & \multicolumn{5}{c}{Estimator type} \\[2mm]
 & $\hat{\var}_{\alpha}^{\textrm{emp}}$ & $\hat{\var}_{\alpha}^{\textrm{norm}}$ & $\hat{\var}_{\alpha}^{\textrm{CF}}$ & $\hat{\var}_{\alpha}^{\textrm{GPD}}$ & $\hat{\var}_{\alpha}^{\textrm{u}}$  \\ 
  \midrule
LoBM.LoOP & 0.071 & 0.073 & 0.067 & 0.067 & 0.069    \\ 
  BM1.OP2 & 0.076 & 0.070 & 0.069 & 0.069 & 0.065    \\ 
  BM1.OP3 & 0.071 & 0.064 & 0.063 & 0.064 & 0.061    \\ 
  BM1.OP4 & 0.069 & 0.071 & 0.067 & 0.067 & 0.068    \\ 
  LoBM.HiOP & 0.071 & 0.071 & 0.070 & 0.067 & 0.068 \\ 
$\cdots$ & $\cdots$ & $\cdots$ & $\cdots$ & $\cdots$   \\
  \midrule
  mean & 0.073 & 0.071 & 0.068 & 0.067 & 0.067   \\ 
  \midrule \\[-5.5mm] \midrule
\multicolumn{1}{p{2.3cm}}{Type of data:}  & \multicolumn{5}{c}{SIMULATED} \\ \midrule
 & $\hat{\var}_{\alpha}^{\textrm{emp}}$ & $\hat{\var}_{\alpha}^{\textrm{norm}}$ & $\hat{\var}_{\alpha}^{\textrm{CF}}$ & $\hat{\var}_{\alpha}^{\textrm{GPD}}$ & $\hat{\var}_{\alpha}^{\textrm{u}}$  \\ 
  \midrule
LoBM.LoOP   &  0.065 & 0.057 & 0.055 & 0.056 & {\bf 0.051}  \\ 
  BM1.OP2   &  0.064 & 0.053 & 0.053 & 0.053 & {\bf 0.050} \\ 
  BM1.OP3   &  0.069 & 0.058 & 0.058 & 0.060 & {\bf 0.052} \\ 
  BM1.OP4   &  0.069 & 0.057 & 0.058 & 0.062 & {\bf 0.053} \\ 
  LoBM.HiOP &  0.060 & 0.054 & 0.053 & 0.056 & {\bf 0.047}\\
$\cdots$ & $\cdots$ & $\cdots$ & $\cdots$ & $\cdots$   \\
  \midrule
  mean &  0.066 & 0.057 & 0.057 & 0.058 & {\bf 0.051}   \\ 
\bottomrule
   \end{tabular}
  \label{t:boot0}
\end{table}



To gain further insight on the performance of the Gaussian unbiased estimator, we have replicated  the results from the simulations in Table~\ref{t:boot0}  for $N=10.000$ times and the first portfolio LoBM.LoOP.  We consider three statistics based on exception rate: for any considered estimator $\hat\rho$ and sample $i \in \{1,\dots,N\}$ we consider first the \emph{exceedance rate} $\textrm{ER}_i(\hat\rho)$, second the 
\emph{relative deviation }
\[
\textrm{RD}_i(\hat \rho):=\frac{\textrm{ER}_i(\hat\rho)-\textrm{ER}_i(\hat{\var}_{\alpha}^{\textrm{u}})}{\textrm{ER}_i(\hat{\var}_{\alpha}^{\textrm{u}})}
\]
and, third, the  \emph{outperformance rate} of the unbiased estimator in the sense that the exceedance rate is closer to $\alpha=0.05$,
\begin{equation}\label{eq:OR}
\textrm{OR}_i(\hat \rho):=
\begin{cases}
1 & \textrm{if } |\textrm{ER}_i(\hat\rho)-\alpha|>|\textrm{ER}_i(\hat{\var}_{\alpha}^{\textrm{u}})-\alpha|,\\
0 & \textrm{otherwise}.
\end{cases}
\end{equation}

In Table~\ref{t:gaussian2} we state mean and standard deviations (sd) of these statistics. It clearly shows that the competing estimators underestimate the risk systematically and exceed the targeted level on average up to 29.2\% more times than the unbiased estimator. Moreover, as could be seen from values of the outperformance rate $\textrm{OR}$, the 
 exception rate fit
of the Gaussian unbiased estimator is better in almost all cases.


\begin{table}[ht]
\caption{We fit a normal distribution to the first portfolio from the Fama \& French dataset, i.e.  LoBM.LoOP portfolio, compare Table \ref{t:boot0}. From this distributions\ we simulate 10.000  samples of size $n=2500$  and perform the standard backtest 10.000 times.
The table presents the average exception (ER) rate for Value-at-Risk at level $\alpha=5\%$. It can be seen that for the biased estimators, the average mean exception rate is significantly higher than the expected rate of $0.05$ while the Gaussian unbiased estimator perform very well. The relative deviation (RD) shows that the exceedance rate for Gaussian unbiased estimator is usually lower in comparison with other estimators, eliminating the effect of risk underestimation. The outperformance rate (OR)  shows that in almost all cases, the exception rate for Gaussian unbiased estimator is closer to 0.05 than the exception rate of any other of the considered estimators.  }
\centering
\begin{tabular}{ll*{3}{c}cc}\toprule  
 Estimator & &  \multicolumn{2}{c}{$\textrm{ER}$} & \multicolumn{2}{c}{$\textrm{RD}$} & $\textrm{OR}$ 
 \\
 &&  mean & sd &  mean & sd & mean 
 \\
 \midrule
Percentile & $\hat{\var}_{\alpha}^{\textrm{emp}}(x)$ & 0.067 & 0.004 & 29.2\% & 8.9\% & 100\%  
\\ 
  Modified C-F & $\hat{\var}_{\alpha}^{\textrm{CF}}(x)$& 0.057 & 0.003 & 11.2\% & 5.0\% & 91.7\% 
\\ 
  Gaussian & $\hat{\var}_{\alpha}^{\textrm{norm}}(x)$ & 0.057 & 0.004 & 9.8\% & 3.0\% & 88.2\% 
\\ 
  GPD & $\hat{\var}_{\alpha}^{\textrm{GPD}}(x)$ & 0.058 & 0.003 & 12.5\% & 6.4\% & 93.3\% 
\\ 
  Gaussian unbiased & $\hat{\var}_{\alpha}^{\textrm{u}}(x)$& 0.052 & 0.003 & - & - & - 
\\ 
\bottomrule
\end{tabular}
\label{t:gaussian2}
\end{table}


\subsection{Backtesting Expected Shortfall}\label{sec5.3}
In this example we will use the same dataset, but instead of $\var$ at level $5\%$ we consider ES at level $10\%$. Following the notation in Equations~\eqref{eq:var.bia}--\eqref{eq:var.norm} we obtain the estimators 

\begin{align}
\hat{\textrm{ES}}_{\alpha}^{\textrm{emp}}(x) &:=-\left(\frac{ \sum_{i=1}^{n}x_i\1_{\{x_i+\hat{\var}_{\alpha}^{\textrm{emp}}(x)<0\}}}{\sum_{i=1}^{n}\1_{\{x_i+\hat{\var}_{\alpha}^{\textrm{emp}}(x)<0\}}} \right),\label{eq:es.hist}\\
\hat{\textrm{ES}}_{\alpha}^{\textrm{CF}}(x) & :=-\left(\bar{x} +\bar{\sigma}(x)C(\bar{Z}^{\alpha}_{CF}(x)) \right),\label{eq:es.mod}\\
\hat{\textrm{ES}}_{\alpha}^{\textrm{norm}}(x) & :=-\left(\bar{x}+\bar{\sigma}(x)\frac{\phi(\Phi^{-1}(\alpha))}{1-\alpha}\right),\label{eq:es.norm}  \\
\hat{\textrm{ES}}_{\alpha}^{\textrm{GPD}}(x) & := \frac{\hat{\var}_{\alpha}^{\textrm{emp}}(x)}{1-\hat\xi} +\frac{\hat\beta-\hat\xi u}{1-\hat\xi}\label{eq:es.gpd},
\end{align}
where $\Phi$ and $\phi$ denotes the cumulative distribution function and the density function of the standard normal distribution, $C(\bar{Z}^{\alpha}_{CF}(x))$ is a standard Cornish-Fisher lower $\alpha$-tail estimator (see~\cite[Equation~(18)]{BouPerCro2008} for details) and {$(u,\hat\beta,\hat\xi)$ is a set of parameters from GPD estimation\footnote{As before, for each portfolio we set the threshold value $u$ to match the $0.7$-empirical quantile of the corresponding sample.} (see Example~\ref{ex:gpd} for details)}. We refer to~\cite{MFE} and \cite{Car2009} for more details and derivation of estimators given in~\eqref{eq:es.hist}--\eqref{eq:es.norm}. Moreover, following Example~\ref{ex:3} we introduce the Gaussian unbiased Expected Shortfall estimator letting
\begin{equation}\label{eq:es.bia}
\hat{\textrm{ES}}_{\alpha}^{\textrm{u}}(x)  :=-\left(\bar{x}-\bar{\sigma}(x)a_n\right),
\end{equation}
where $a_n=b_n\sqrt{\nicefrac{(n-1)(n+1)}{n}}$ and $b_n$ is given as a solution of Equation~\eqref{eq:est.tvar.cond}. By doing simple calculations, for $\alpha=10\%$ and $n=50$ we get the approximate value of $a_n$ equal to $-1.81033$.\footnote{We have used a sample of size $10.000.000$ to approximate this value.} Note that the non-elicitability of ES is directly reflected in \eqref{eq:es.hist} and \eqref{eq:es.gpd}. The joint elicitability of ES together with $\var$ is also visible: the estimator for the $\textrm{ES}$ also makes use of an estimator for $\var$.\footnote{An illustrative and self-explanatory example of this phenomena is joint elicitability of mean and variance: the mean estimator is usually  embedded in the variance estimator; see~\cite{LamPenSho2008} for details.}

For the backtest we follow  {\it Test 2} suggested in~\cite{Acerbi2014Risk}. We utilize the $50$ separate subsets of our data denoted by $(x^k_1,\dots,x^k_{50})$, for $k=1,2,\ldots,50$. 
We consider one of the above estimators of expected shortfall and denote generically by $\hat{\textrm{ES}}^{i}_{\alpha}$ the estimation resulting from the $k$-th subset of our data and by $\hat{\var}_{\alpha}^{i}$ the associated $\var_{\alpha}$ estimator obtained using $k$-th sample, both with level $\alpha$. The test statistic for the backtest is given by
\begin{align}\label{Z}
Z :=\frac{1}{49}\sum_{k=1}^{49}\left(\frac{1}{50}\sum_{j=1}^{50}\frac{x^{k+1}_j\, \ind{x^{k+1}_j+\hat \var^{k}_{\alpha}<0 }}{\alpha \,\hat{\textrm{ES}}^{k}_{\alpha}}\right)+1,
\end{align}
see Equation (6) in \cite{Acerbi2014Risk}. The null-hypothesis that the tail distributions coincide below the $\alpha$-quantile corresponds to a vanishing expectation of $Z$. The alternative that expected shortfall or value-at-risk are underestimated corresponds to negative values of the test statistic $Z$.

The results of our backtest are presented in Table~\ref{t:boot2}. An estimation of the bias as outlined in Section \ref{sec:backtesting} leads to quite similar results and is not presented here. The unbiased estimator clearly outperforms the biased estimators, both on the market data and the simulated data.

\begin{remark}
The at first sight surprising  performance of the unbiased estimator can be traced back to the following heuristics: the test statistic introduced in~\eqref{Z} is an empirical counterpart of the value
\begin{align}\label{temp856}
\bE_{\theta}\left[\frac{X \ind{X+\var_{\alpha}<0}}{\alpha \hat{\textrm{ES}}_{\alpha}}+1\right],
\end{align}
following similar arguments as laid out in Section \ref{sec:backtesting}.
If we assume continuity of the underlying distribution, we obtain
\[
\eqref{temp856} =\bE_{\theta}\left[\frac{X}{\hat{\textrm{ES}}_{\alpha}}+1\,|\, X+\var_{\alpha}<0 \right]=-\rho_{\theta}\left({\frac{X}{\hat{\textrm{ES}}}_{\alpha}}+1\right)=-\rho_{\theta}\left({\frac{X+\hat{\textrm{ES}}_{\alpha}}{\hat{\textrm{ES}}_{\alpha}}}\right),
\]
where $\rho_{\theta}$ denotes $\textrm{ES}_{\alpha}$ under the unknown true parameter $\theta$.
Recall that $\rho_{\theta}(X+\hat{\textrm{ES}}_{\alpha})=0$ for the unbiased estimator. If the multiplication by $\hat{\textrm{ES}}_{\alpha}^{-1}$ is compatible with positive homogeneity \footnote{Or, similarily, assuming that $\hat{\textrm{ES}}_{\alpha}$ is close to a constant and positve.}, this suggests that also \eqref{temp856} should be close to zero. 
\end{remark}

To further substantiate this,  we have replicated the results from the simulations in Table~\ref{t:boot2}  for the empirical estimator \eqref{eq:es.hist}, standard Gaussian estimator \eqref{eq:es.hist}, and Gaussian unbiased estimator \eqref{eq:es.bia} applied to the first portfolio LoBM.LoOP for $N=10.000$ times. The obtained mean values of the statistic (with standard errors in parenthesis) are -0.179 (0.040), -0.104 (0.042), and -0.031 (0.042), respectively. Also, we can consider the analogue of {\it outperformance rate} statistic given in \eqref{eq:OR} for expected shortfall, with $0$ as the reference value. In 95.4\% of cases, the value of test statistics for the Gaussian unbiased estimator of ES was closer to zero compared to the value of the test statistics for the standard Gaussian ES estimator. For the empirical estimator, the outperformance rate was equal to 100\%. This shows that the results from Table~\ref{t:boot2} are statistically significant.

\begin{table}[ht]
\caption{We run the backtest on the same data taken for Table \ref{t:boot0}. The sample is split into intervals of length 50 and on each subset the estimation is performed. The table presents the value of the backtesting statistic $Z$ given in Equation \eqref{Z} at level $\alpha=10\%$ for 
the empirical sample quantile ${\textrm{ES}}_{\alpha}^{\textrm{emp}}$, the  modified Cornish-Fisher estimator ${\textrm{ES}}_{\alpha}^{\textrm{CF}}$ the classical Gaussian estimator ${\textrm{ES}}_{\alpha}^{\textrm{norm}}$, the GPD estimator ${\textrm{ES}}_{\alpha}^{\textrm{GPD}}$, and the  unbiased estimator  ${\textrm{ES}}_{\alpha}^{\textrm{u}}$.
The row labelled "mean" shows  the mean of all  numbers (25 - here we include all portfolios) in a given column, where we would expect $0$ if the estimator performs correctly. Negative values of $Z$ correspond to underestimation of risk. For the Gaussian data,  the average value of $Z$  for the biased estimators is significantly lower than for the unbiased estimators, suggesting a good performance of the latter ones (values of the unbiased estimators in bold type).}
\centering
\begin{tabular}{lrrrrr}\toprule
\multicolumn{1}{p{2.3cm}}{Type of data:}  & \multicolumn{5}{c}{MARKET} \\ \midrule
Portfolio & \multicolumn{5}{c}{Estimator type} \\[2mm]
 & ${\textrm{ES}}_{\alpha}^{\textrm{emp}}$ & ${\textrm{ES}}_{\alpha}^{\textrm{norm}}$ & ${\textrm{ES}}_{\alpha}^{\textrm{CF}}$ & ${\textrm{ES}}_{\alpha}^{\textrm{GPD}}$ & ${\textrm{ES}}_{\alpha}^{\textrm{u}}$  \\ 
  \midrule
 LoBM.LoOP &  -0.357 & -0.393 & -0.325 & -0.302 & -0.331    \\ 
   BM1.OP2 &  -0.428 & -0.303 & -0.338 & -0.335 & -0.235    \\ 
   BM1.OP3 &  -0.327 & -0.322 & -0.336 & -0.295 & -0.254    \\ 
   BM1.OP4 &  -0.326 & -0.354 & -0.348 & -0.282 & -0.272    \\ 
 LoBM.HiOP &  -0.424 & -0.421 & -0.371 & -0.335 & -0.331    \\ 
 $\cdots$ & $\cdots$ & $\cdots$ & $\cdots$ & $\cdots$   \\
  \midrule
  mean & -0.374 & -0.363 & -0.339 &  -0.308 & -0.290  \\ 
  \midrule \\[-5.5mm] \midrule
\multicolumn{1}{p{2.3cm}}{Type of data:}  & \multicolumn{5}{c}{SIMULATED} \\ \midrule
 & ${\textrm{ES}}_{\alpha}^{\textrm{emp}}$ & ${\textrm{ES}}_{\alpha}^{\textrm{norm}}$ & ${\textrm{ES}}_{\alpha}^{\textrm{CF}}$ & ${\textrm{ES}}_{\alpha}^{\textrm{GPD}}$ & ${\textrm{ES}}_{\alpha}^{\textrm{u}}$  \\ 
  \midrule
LoBM.LoOP   & -0.177 & -0.073 & -0.077 & -0.104 & {\bf -0.005}  \\ 
  BM1.OP2   & -0.143 & -0.083 & -0.069 & -0.074 & {\bf -0.014}  \\ 
  BM1.OP3   & -0.220 & -0.084 & -0.100 & -0.157 & {\bf -0.019}  \\ 
  BM1.OP4   & -0.224 & -0.086 & -0.101 & -0.150 & {\bf -0.012}  \\ 
  LoBM.HiOP & -0.183 & -0.082 & -0.072 & -0.098 & {\bf -0.016}  \\
  $\cdots$ & $\cdots$ & $\cdots$ & $\cdots$ & $\cdots$   \\
  \midrule
  mean & -0.174 & -0.101 & -0.103 & -0.109 & {\bf -0.030}  \\ 
\bottomrule
   \end{tabular}
  \label{t:boot2}
\end{table}

\vspace{2mm}

We want to emphasize that in this small empirical study we only considered $\var$ estimators defined in Equations \eqref{eq:var.bia}-\eqref{eq:var.norm} and \eqref{eq:var.gpd} as well as their equivalents for ES. The results presented  could possibly be improved using estimation combined with e.g. GARCH filtering, as is sometimes done in practice \cite{SoYu2006,BerObr2002}. Furthermore, other risk measures  such as the median shortfall (median tail loss) \cite{KouPenHey2013} and other  backtesting procedures could be considered  (see e.g.~\cite{Acerbi2014Risk,Fissler2016}). This, however, is beyond the scope of this paper.

\subsection{Estimator performance and consistent scoring functions}\label{sec:nonstandardbacktesting}

In this section we take up the critical remarks from \cite{Gneiting2011} on backtesting procedures as discussed in Section \ref{sec:elicitabilty} and confirm the results from the previous section with theoretically substantiated 
procedures based on consistent scoring functions. While in the previous section we were interested in the estimation process conservativeness, here we compare the estimators and check their fit by exploiting the concept of {\it forecast dominance} and {\it comperativeness}; see Section 2.3 of \cite{NolZie2017}.

We utilise the data and the framework from the previous section: Let us assume we are given a dataset of length 2500. For $k\in\{1,\ldots,50\}$, let $(x^{k}_1,\dots,x^{k}_{50})$ denote the $k$-th subset of the data and let $\hat\rho^{k}$ denote the value of the associated $\hat\rho$ estimator (i.e. $\var_{\alpha}$ or $\textrm{ES}_{\alpha}$ estimator) obtained using the $k$-th subset. Given a scoring function $S$ we define the mean score statistic
\begin{equation}\label{eq:mean.score1}
\bar S(\hat\rho)= \frac{1}{49}\sum_{k=1}^{49}\left(\frac{1}{50}\sum_{j=1}^{50}S(-\hat\rho^{k},x^{k+1}_j)\right).
\end{equation}
If the scoring function is consistent, the mean score value $\bar S$ could be treated as a performance criterion -- the smaller the value, the better the estimator. In particular, given two competing estimators one could compare their mean score to determine which one is better; see \cite{NolZie2017}.
For $\var$ performance test, following \cite{Fissler2016}, we use the (consistent) scoring function given in~\eqref{eq:scor.cons.var}. The results for market and simulated data are presented in Table~\ref{t:score1}. 

\begin{table}[ht]
\caption{We run the scoring performance test on the same data taken for Table \ref{t:boot0}. The sample is split into intervals of length 50 and on each subset the estimation is performed. The table presents the value of the mean score $\bar S$ given in Equation \eqref{eq:mean.score1}, multiplied by $1,000$, for 
the empirical sample quantile $\hat{\var}_{\alpha}^{\textrm{emp}}$, the  modified Cornish-Fisher estimator $\hat{\var}_{\alpha}^{\textrm{CF}}$, the GPD estimator $\hat{\var}_{\alpha}^{\textrm{GPD}}$ and the classical Gaussian estimator $\hat{\var}_{\alpha}^{\textrm{norm}}$.
The row labelled "mean" shows  the mean of all  numbers (25 - here we include all portfolios) in a given column. 
On the simulated data, the Cornish-Fisher and the GPD estimator show a lower statistic than the unbiased estimator which is due to the normality assumption in $\hat{\var}_{\alpha}^{\textrm{u}}$. 
For the Gaussian data, in most of the cases, the values of $\bar S$ for the biased estimators are higher than for the unbiased estimators, suggesting a good performance of the latter ones (values of the unbiased estimators in bold type).}
\centering
\begin{tabular}{lrrrrr}\toprule
\multicolumn{1}{p{2.3cm}}{Type of data:}  & \multicolumn{5}{c}{MARKET} \\ \midrule
Portfolio & \multicolumn{5}{c}{Estimator type} \\[2mm]
 & $\hat{\var}_{\alpha}^{\textrm{emp}}$ & $\hat{\var}_{\alpha}^{\textrm{norm}}$ & $\hat{\var}_{\alpha}^{\textrm{CF}}$ & $\hat{\var}_{\alpha}^{\textrm{GPD}}$ & $\hat{\var}_{\alpha}^{\textrm{u}}$  \\ 
  \midrule
LoBM.LoOP & 2.005 & 2.019 & 1.992 & 1.984 & 2.012 \\ 
  BM1.OP2 & 1.964 & 1.948 & 1.937 & 1.910 & 1.944 \\ 
  BM1.OP3 & 1.692 & 1.693 & 1.669 & 1.677 & 1.691 \\ 
  BM1.OP4 & 1.389 & 1.402 & 1.383 & 1.383 & 1.399 \\ 
  LoBM.HiOP & 1.335 & 1.350 & 1.333 & 1.342 & 1.347 \\ 
  $\cdots$ & $\cdots$ & $\cdots$ & $\cdots$ & $\cdots$   \\
  \midrule
  mean & 1.780 & 1.770 & 1.761 & 1.761 & 1.766  \\ 
  \midrule \\[-5.5mm] \midrule
\multicolumn{1}{p{2.3cm}}{Type of data:}  & \multicolumn{5}{c}{SIMULATED} \\ \midrule
 & $\hat{\var}_{\alpha}^{\textrm{emp}}$ & $\hat{\var}_{\alpha}^{\textrm{norm}}$ & $\hat{\var}_{\alpha}^{\textrm{CF}}$ & $\hat{\var}_{\alpha}^{\textrm{GPD}}$ & $\hat{\var}_{\alpha}^{\textrm{u}}$  \\ 
  \midrule
LoBM.LoOP & 1.898 & 1.859 & 1.867 & 1.874 & {\bf 1.857} \\ 
  BM1.OP2 & 1.938 & 1.898 & 1.901 & 1.904 & {\bf 1.899} \\ 
  BM1.OP3 & 1.599 & 1.557 & 1.558 & 1.574 & {\bf 1.554} \\ 
  BM1.OP4 & 1.292 & 1.245 & 1.254 & 1.272 & {\bf 1.244} \\ 
  LoBM.HiOP & 1.195 & 1.177 & 1.182 & 1.188 & {\bf 1.177} \\ 
  $\cdots$ & $\cdots$ & $\cdots$ & $\cdots$ & $\cdots$   \\
  \midrule
  mean & 1.714 & 1.684 & 1.690 & 1.698 & {\bf 1.683} \\ 
\bottomrule
   \end{tabular}
  \label{t:score1}
\end{table}

For the ES, the situation is more complex, as we need to consider the joint elicitability with $\var$. Following~\cite{Fissler2016} we use the (joint) $\var$-ES consistent scoring function
\begin{equation}\label{eq:score.fun.2}
S(x_1,x_2,y)=(\1_{\{x_1\geq y\}}-\alpha)(x_1-y)+\frac{1}{\alpha}\frac{e^{x_2}}{1+e^{x_2}}\1_{\{x_1\geq y\}}(x_1-y)+\frac{e^{x_2}}{1+e^{x_2}}(x_2-x_1)-\frac{e^{x_2}}{1+e^{x_2}}.
\end{equation}
As in Section~\ref{sec5.3}, for each ES estimator we use the associated $\var$ estimator. The results for market and simulated data for the corresponding mean score are presented in Table~\ref{t:score3}.

Repeated simulations of the above results show statistical significance of the differences,  both for $\var$ and ES (results not shown).

\begin{table}[htp!]
\caption{We run the joint $\var$-ES scoring performance test on the same data taken for Table \ref{t:boot0}. The sample is split into intervals of length 50 and on each subset the estimation is performed. The table presents the value of the mean score with score function given in Equation \eqref{eq:score.fun.2} for the empirical sample quantile ${\textrm{ES}}_{\alpha}^{\textrm{emp}}$, the  modified Cornish-Fisher estimator ${\textrm{ES}}_{\alpha}^{\textrm{CF}}$ the classical Gaussian estimator ${\textrm{ES}}_{\alpha}^{\textrm{norm}}$, the GPD estimator ${\textrm{ES}}_{\alpha}^{\textrm{GPD}}$, and the  unbiased estimator  ${\textrm{ES}}_{\alpha}^{\textrm{u}}$. In each case, the associated $\var$ estimator is used. The row labelled "mean" shows  the mean of all numbers (25 - here we include all portfolios) in a given column. For the Gaussian data, in most of the cases, the values of mean score for the biased estimators are lower than for the unbiased estimators, suggesting a good performance of the latter ones (values of the unbiased estimators in bold type).}
\centering
\begin{tabular}{lrrrrr}\toprule
\multicolumn{1}{p{2.3cm}}{Type of data:}  & \multicolumn{5}{c}{MARKET} \\ \midrule
Portfolio & \multicolumn{5}{c}{Estimator type} \\[2mm]
 & ${\textrm{ES}}_{\alpha}^{\textrm{emp}}$ & ${\textrm{ES}}_{\alpha}^{\textrm{norm}}$ & ${\textrm{ES}}_{\alpha}^{\textrm{CF}}$ & ${\textrm{ES}}_{\alpha}^{\textrm{GPD}}$ & ${\textrm{ES}}_{\alpha}^{\textrm{u}}$  \\ 
  \midrule
LoBM.LoOP & -0.4882 & -0.4879 & -0.4880 & -0.4883 & -0.4881 \\ 
  BM1.OP2 & -0.4886 & -0.4885 & -0.4886 & -0.4888 & -0.4887 \\ 
  BM1.OP3 & -0.4900 & -0.4900 & -0.4899 & -0.4900 & -0.4902 \\ 
  BM1.OP4 & -0.4917 & -0.4918 & -0.4918 & -0.4919 & -0.4919 \\ 
  LoBM.HiOP & -0.4920 & -0.4921 & -0.4921 & -0.4921 & -0.4922 \\ 
  $\cdots$ & $\cdots$ & $\cdots$ & $\cdots$ & $\cdots$   \\
  \midrule
  mean & -0.4895 &-0.4895 &-0.4896 &-0.4897 &-0.4897 \\
  \midrule \\[-5.5mm] \midrule
\multicolumn{1}{p{2.3cm}}{Type of data:}  & \multicolumn{5}{c}{SIMULATED} \\ \midrule
 & ${\textrm{ES}}_{\alpha}^{\textrm{emp}}$ & ${\textrm{ES}}_{\alpha}^{\textrm{norm}}$ & ${\textrm{ES}}_{\alpha}^{\textrm{CF}}$ & ${\textrm{ES}}_{\alpha}^{\textrm{GPD}}$ & ${\textrm{ES}}_{\alpha}^{\textrm{u}}$  \\ 
  \midrule
LoBM.LoOP & -0.4886 & -0.4888 & -0.4888 & -0.4889 & {\bf -0.4891} \\ 
  BM1.OP2 & -0.4886 & -0.4889 & -0.4888 & -0.4889 & {\bf -0.4891} \\ 
  BM1.OP3 & -0.4902 & -0.4906 & -0.4905 & -0.4905 & {\bf -0.4908} \\ 
  BM1.OP4 & -0.4922 & -0.4925 & -0.4924 & -0.4923 & {\bf -0.4927} \\ 
  LoBM.HiOP & -0.4928 & -0.4929 & -0.4928 & -0.4929 & {\bf -0.4931} \\ 
  $\cdots$ & $\cdots$ & $\cdots$ & $\cdots$ & $\cdots$   \\
  \midrule
  mean & -0.4896 & -0.4899 & -0.4898 & -0.4898 & {\bf -0.4901} \\ 
\bottomrule
   \end{tabular}
  \label{t:score3}
\end{table}

\section{Conclusion}\label{sec:conclusion}
In this article, we studied the estimation of risk, with a particular view on \emph{unbiased}  estimators and backtesting. 
The new notion of unbiasedness introduced is motivated from economic principles rather than from statistical reasoning, which links this concept to a better performance in backtesting. Some unbiased estimators, for example the unbiased estimator for value-at-risk in the Gaussian case, can be computed in closed form while for many other cases numerical methods are available. A small empirical analysis underlines the excellent performance of the unbiased estimators with respect to standard backtesting measures. 

%


 \end{document}